\pdfoutput=1
\documentclass[11pt]{article}

\usepackage[letterpaper,margin=1in]{geometry}
\usepackage[T1]{fontenc}
\usepackage[utf8]{inputenc}
\usepackage{lmodern}
\usepackage{amsmath,amssymb,amsthm,mathtools}
\usepackage{bm}
\usepackage{microtype}
\usepackage{booktabs}
\usepackage{graphicx}
\usepackage{placeins}
\usepackage{tikz}
\usetikzlibrary{arrows.meta}
\usepackage{caption}
\usepackage{enumitem}
\usepackage[round,authoryear]{natbib}
\usepackage{xcolor}
\usepackage{hyperref}
\usepackage[nameinlink,noabbrev]{cleveref}

\hypersetup{
  colorlinks=true,
  linkcolor=blue!45!black,
  citecolor=blue!45!black,
  urlcolor=black,
  pdftitle={Polylogarithmic Sparsity of Randomly Reweighted NPMLEs for
    Gaussian Mixtures},
  pdfauthor={Hansheng Jiang}
}
\newcommand{\leanrepo}{\url{https://github.com/hanshengjiang/reweighted-npmle-lean}}

\setlist[itemize]{leftmargin=1.4em,itemsep=0.15em,topsep=0.3em}
\setlist[enumerate]{leftmargin=1.7em,itemsep=0.15em,topsep=0.3em}
\allowdisplaybreaks[2]
\newtheorem{theorem}{Theorem}
\newtheorem{proposition}{Proposition}
\newtheorem{lemma}{Lemma}
\newtheorem{corollary}{Corollary}
\theoremstyle{definition}
\newtheorem{definition}{Definition}
\newtheorem{assumption}{Assumption}
\theoremstyle{remark}

\DeclareMathOperator*{\argmax}{arg\,max}

\DeclareMathOperator{\conv}{conv}
\DeclareMathOperator{\supp}{supp}
\DeclareMathOperator{\rank}{rank}
\DeclareMathOperator{\diag}{diag}
\DeclareMathOperator{\KL}{KL}
\DeclareMathOperator{\TV}{TV}
\newcommand{\Pcal}{\mathcal P}
\newcommand{\Mcal}{\mathcal M}
\newcommand{\Ccal}{\mathcal C}
\newcommand{\Acal}{\mathcal A}
\newcommand{\Ucal}{\mathcal U}
\newcommand{\Ncal}{\mathcal N}
\newcommand{\RR}{\mathbb R}
\newcommand{\EE}{\mathbb E}
\newcommand{\PP}{\mathbb P}
\newcommand{\one}{\bm 1}
\newcommand{\norm}[1]{\lVert #1\rVert}
\newcommand{\abs}[1]{\lvert #1\rvert}

\newcommand{\vhat}{\widehat v}
\newcommand{\Ghat}{\widehat G}
\newcommand{\Gtil}{\widetilde G}
\newcommand{\vtil}{\widetilde v}
\newcommand{\eps}{\varepsilon}
\newcommand{\sext}{s_{\mathrm{ext}}}

\title{\bfseries Polylogarithmic Sparsity of Randomly Reweighted\\
NPMLEs for Gaussian Mixtures\thanks{The proofs were developed with the assistance of GPT-6 Pro.  The author formulated the problem, directed the approach, checked and rewrote the AI-generated arguments, and takes full responsibility for the correctness and presentation of the results.  The numerical study and the Lean 4 formalization serve as independent checks on the results.  The formalization covers the full proof of the main theorem and is available at \leanrepo.}}
\author{\normalsize Hansheng Jiang\\
\small University of Toronto\\
\small \href{mailto:hansheng.jiang@utoronto.ca}{\texttt{hansheng.jiang@utoronto.ca}}}
\date{}

\newcommand{\numReportedFits}{4160}
\newcommand{\numPairedDatasets}{920}

\newcommand{\numMaxResidual}{8.03\times10^{-7}}
\newcommand{\numMaxResidualOneD}{4.36\times10^{-7}}
\newcommand{\numMultistartDistanceOneD}{3.9\times10^{-16}}
\newcommand{\numMultistartDistanceTwoD}{1.0\times10^{-14}}
\newcommand{\numBalancedRiskMin}{1.001}
\newcommand{\numBalancedRiskMax}{1.020}
\newcommand{\numBalancedRiskAverage}{1.013}
\newcommand{\numBalancedDistanceMin}{8.31\times10^{-6}}
\newcommand{\numBalancedDistanceMax}{1.76\times10^{-4}}
\newcommand{\numBalancedSupportAgreement}{90.3}
\newcommand{\numBalancedSupportDifferenceMin}{-0.125}
\newcommand{\numBalancedSupportDifferenceMax}{0.025}
\newcommand{\numBalancedLossMin}{0.0345}
\newcommand{\numBalancedLossMax}{0.0964}
\newcommand{\numBalancedScaledLossMin}{0.253}
\newcommand{\numBalancedScaledLossMax}{0.532}
\newcommand{\numBalancedScaledLogfitMin}{0.479}
\newcommand{\numBalancedScaledLogfitMax}{0.928}
\newcommand{\numTinyRiskMin}{0.999985}
\newcommand{\numTinyRiskMax}{1.000017}
\newcommand{\numTinyDistanceMin}{8.03\times10^{-14}}
\newcommand{\numTinyDistanceMax}{1.32\times10^{-9}}
\newcommand{\numTinySupportAgreement}{100.0}
\newcommand{\numBBRiskMin}{1.60}
\newcommand{\numBBRiskMax}{2.19}
\newcommand{\numBBLossMin}{2.14}
\newcommand{\numBBLossMax}{7.86}
\newcommand{\numPathBalancedAlpha}{59.76}
\newcommand{\numSlopeWeight}{-0.503}
\newcommand{\numSlopeDistance}{-1.016}
\newcommand{\numSlopeLoss}{-1.016}
\newcommand{\numSlopeLogfit}{-1.012}
\newcommand{\numPathOrdinaryRisk}{0.002662}
\newcommand{\numPathBalancedDistance}{4.40\times10^{-5}}
\newcommand{\numPathBalancedLoss}{0.0453}
\newcommand{\numPathBalancedRisk}{0.002664}
\newcommand{\numPathBalancedSupportSummary}{28 draws have 3 atoms and 2 draws have 4 atoms}
\newcommand{\numPathBalancedMedianSupport}{3}

\newcommand{\numPathTinyDistance}{1.52\times10^{-11}}
\newcommand{\numPathTinyLoss}{1.65\times10^{-8}}

\newcommand{\numPathTinySupportSummary}{30 draws have 3 atoms}
\newcommand{\numPathTinyMedianSupport}{3}

\newcommand{\numPathBBDistance}{1.95\times10^{-3}}
\newcommand{\numPathBBLoss}{2.03}

\newcommand{\numPathBBUpperSupport}{4}
\newcommand{\numPathBalancedWasserstein}{0.0308}
\newcommand{\numNegativeRawLossCount}{6}
\newcommand{\numMinRawLoss}{-8.87\times10^{-6}}

\begin{document}
\maketitle

\begin{abstract}
The nonparametric maximum likelihood estimator (NPMLE) of a Gaussian location mixture maximizes the likelihood over the infinite-dimensional space of mixing distributions.  The maximizing mixing distribution can be nonunique, and the classical bound on its number of atoms grows linearly with the sample size $n$.  We show that a vanishingly small random perturbation of the likelihood yields exact polylogarithmic sparsity.  The resulting randomly reweighted NPMLE maximizes a weighted likelihood whose independent weights, taken to be Gamma in our analysis, concentrate around one as $n$ grows.  With high probability, it is unique, has $O\{(\log n/\log\log n)^d+\log n\}$ atoms in dimension $d$, nearly maximizes the ordinary likelihood, and estimates the mixture density at a Hellinger rate that is parametric up to logarithmic factors.  This sparsity holds for the estimator itself, not for an approximation of it, and requires no support penalty.  The proof rests on an effective-dimension principle for positive kernel mixtures: low-dimensional variation of the fitted values controls the support of every extreme point of the set of maximizers.  Numerical illustrations verify that the reweighted NPMLE has Hellinger risk and support size comparable to those of the ordinary NPMLE.
\end{abstract}

\section{Introduction}\label{sec:introduction}

Let $X_1,\ldots,X_n\in\RR^d$ be independent observations from a Gaussian location mixture.  For a mixing distribution $G$ and the standard Gaussian density $\phi_d$ on $\RR^d$, write
\[
  f_G(x)=\int \phi_d(x-\theta)\,dG(\theta).
\]
Given a compact parameter set $K\subset\RR^d$, let $\Pcal(K)$ be the set of all probability measures on $K$.  The nonparametric maximum likelihood estimator (NPMLE) of \citet{kieferwolfowitz1956} is any maximizer of the log-likelihood,
\[
  \Ghat_0\in\argmax_{G\in\Pcal(K)}\sum_{i=1}^n\log f_G(X_i).
\]
Several features make the NPMLE attractive: it requires neither a prespecified number of mixture components nor a tuning parameter, and, because the log-likelihood is concave in $G$ and $\Pcal(K)$ is convex, it is defined by a convex optimization problem \citep{lindsay1983geometry,koenkermizera2014}.  These features have made the NPMLE a standard tool for empirical Bayes inference, with applications ranging from count, survival and longitudinal data \citep{koenkergu2017} to ranking, selection and a broad range of problems in economics \citep{gukoenker2023,walters2024,chen2026}.  The problem is posed over an infinite-dimensional set, however, and its definition says little about the structure of a maximizing mixing distribution.

Classical convex geometry guarantees the existence of a maximizer with at most $n$ atoms, and all maximizers have the same fitted vector $(f_{\Ghat_0}(X_1),\ldots,f_{\Ghat_0}(X_n))$ \citep{lindsay1983geometry}.  This deterministic $n$-point bound does not explain the much smaller models often returned in practice.  For the unrestricted problem in one dimension, the maximizing mixing distribution is unique \citep{lindsay1983partii}, and \citet{polyanskiywu2020} proved that it has $O(\log n)$ atoms with high probability when the true mixing distribution is sub-Gaussian, a phenomenon they called self-regularization.  In several dimensions, \citet{wang2026} shows that every unrestricted Gaussian NPMLE has finite support, and \citet{wang2026uniqueness} establishes uniqueness for Lebesgue-almost every dataset.  These results do not provide a polylogarithmic bound on the number of atoms.  Moreover, whereas in one dimension the number of atoms is bounded by a function of the range of the data \citep{polyanskiywu2020}, \citet[Proposition~4.2]{polyanskiysellke2025} construct uniformly bounded datasets in every dimension $d\ge2$ for which the number of atoms of every NPMLE diverges with $n$.

Beyond these structural properties, Gaussian-mixture NPMLEs have been studied extensively as estimators of mixing distributions, marginal densities, and empirical Bayes rules; see, among others, \citet{laird1978}, \citet{koenkermizera2014}, \citet{zhang2009}, \citet{sahaguntuboyina2020}, \citet{soloffetal2025}, and \citet{chendebignatiadis2026}.  The resulting density-estimation rates can be parametric up to logarithmic factors, and recent entropy characterizations clarify the minimax benchmark \citep{jiaetal2023}.  These results establish the statistical value of the NPMLE, but neither they nor the structural results above provide a polylogarithmic support bound for an exact ordinary NPMLE in several dimensions.

We establish such a bound for an asymptotically negligible random reweighting of the likelihood, rather than for the ordinary NPMLE $\Ghat_0$ itself.  Given a concentration parameter $\alpha_n>0$, draw weights
\[
 W_1,\ldots,W_n\stackrel{\mathrm{iid}}\sim\mathrm{Gamma}(\alpha_n,\alpha_n),
\]
independently of the data, where shape and rate are both $\alpha_n$, and define the randomly reweighted NPMLE
\begin{equation}\label{eq:rrnpmle-intro}
 \Ghat_W\in\argmax_{G\in\Pcal(K)}
       \sum_{i=1}^n W_i\log f_G(X_i).
\end{equation}
The weights have mean one and variance $1/\alpha_n$, and we let $\alpha_n\to\infty$, so that the weights converge to one and the reweighted log-likelihood approaches the ordinary one.  Computing $\Ghat_W$ is no harder than computing $\Ghat_0$, because the weights enter the same convex program only through the coefficients of the log-likelihood.  Since multiplying all weights by a common constant does not change the maximizers, $\Ghat_W$ depends on the weights only through $(W_i/\sum_jW_j)_{i\le n}$, which has the $\mathrm{Dirichlet}(\alpha_n,\ldots,\alpha_n)$ distribution.  Dirichlet-weighted likelihoods are classical in the Bayesian bootstrap of \citet{rubin1981} and the weighted likelihood bootstrap of \citet{newtonraftery1994}, where repeated weight draws produce a bootstrap or approximate posterior distribution.  Here, instead, a single, increasingly concentrated draw of the weights defines a point estimator.

Our main result, \Cref{thm:main}, shows that this random reweighting selects a unique maximizer that is exactly sparse: the estimator itself, and not merely an approximation of it, has polylogarithmically many atoms.  The distinction matters because mixing distributions with polylogarithmically many atoms whose log-likelihood is within a polynomially small amount of the maximum are easy to construct by moment matching, as in the proof of \Cref{lem:finite-net}; the difficulty lies in controlling an exact maximizer.  For a fixed compact parameter set and fixed dimension, set
\[
 r_n=\left(\frac{\log n}{\log\log n}\right)^d,
 \qquad \bar r_n=r_n+\log n,
 \qquad \alpha_n=\bar r_n\log n.
\]
Then, for every fixed $b>0$, with probability at least $1-C_bn^{-b}$, the reweighted NPMLE is unique and has at most $C_b\bar r_n$ atoms.  It is also close to the ordinary NPMLE: the weight perturbation satisfies $\max_i\abs{W_i-1}\le C_b\bar r_n^{-1/2}$, and both the ordinary log-likelihood gap and the log-fit discrepancy, the squared Euclidean distance between the logarithms of its fitted vector and of the ordinary fitted vector, are at most $C_b/\log n$.  The balanced concentration is convenient rather than necessary.  For any $L>0$, the choice $\alpha_n=n^{2L+2}$ makes the weight perturbation at most $n^{-L}$, yet the reweighted NPMLE remains unique with at most $C_{b,L}\bar r_n$ atoms, again with probability at least $1-C_bn^{-b}$ (\Cref{cor:tiny}).

Compared with the classical bound of $n$ atoms \citep{lindsay1983geometry}, the bound $C_b\bar r_n$ is polylogarithmic in every fixed dimension.  In one dimension, where the unrestricted ordinary NPMLE is already unique \citep{lindsay1983partii} and has $O(\log n)$ atoms \citep{polyanskiywu2020}, it is also $O(\log n)$.  In several dimensions, our bound is $O\{(\log n/\log\log n)^d\}$, and, with high probability, the reweighted NPMLE over the fixed compact set $K$ is unique.  The sparsity mechanism is not specific to Gaussian kernels.  For any positive continuous kernel, \Cref{thm:effective} shows that if all fitted vectors, measured relative to the ordinary fitted vector, lie close to an $r$-dimensional subspace, then with high probability every extreme point of the set of maximizers of the weighted likelihood, which we call an extreme optimizer, has $O(r+\log n)$ atoms; the argument uses the randomness of the weights in an essential way.  For Gaussian kernels, analyticity gives $r=O(r_n)$ (\Cref{prop:gaussian-width}), and a generic linear independence property of Gaussian evaluation vectors turns sparsity into uniqueness (\Cref{prop:generic}).  The bounds on extreme optimizers and on the log-likelihood gap hold for every dataset of radius $O(\sqrt{\log n})$ and do not require the model to be correctly specified (\Cref{cor:conditional}).

A sparse and unique maximizer is useful only if it is also statistically accurate.  For the ordinary NPMLE, the squared Hellinger loss is known to be of order at most $(\log n)^2/n$ in one dimension \citep{zhang2009} and $(\log n)^{d+1}/n$ in $d$ dimensions \citep{sahaguntuboyina2020}, and the minimax rate for this class is characterized through its metric entropy \citep{jiaetal2023}.  The reweighted NPMLE also satisfies a Hellinger bound that is parametric up to logarithmic factors.  Under correct specification,
\[
 H^2(f_{\Ghat_W},f_{G^*})
 =O_{\PP}\!\left(\frac{(\log n)^{d+1}}
 {n(\log\log n)^d}\right),
\]
and the same bound holds when $\alpha_n=n^{2L+2}$.  It follows from a uniform bound of the same type as those just cited: with high probability, the bound holds simultaneously for every mixing distribution whose log-likelihood is within one of that of the true mixing distribution $G^*$ (\Cref{thm:near-mle}), and the reweighted NPMLE is such a mixing distribution because its ordinary log-likelihood gap is $O(1/\log n)$.  Relative to the rates above, the extra factor $(\log\log n)^{-d}$ comes from a moment-matching entropy bound for compactly supported mixing distributions, so it applies equally to the ordinary NPMLE.  The same argument covers approximate numerical maximizers of the reweighted likelihood (\Cref{cor:approx}) and gives weak consistency of $\Ghat_W$ for $G^*$ (\Cref{cor:weak}).  The logarithmic factors are not claimed to be optimal; \Cref{sec:hellinger} compares this bound with sharper one-dimensional results.

The main technical device behind the sparsity bound is a change of variables for the weights.  Write $\vhat(w)$ for the fitted vector shared by all maximizers of the likelihood with weights $w$.  An extreme optimizer with $k$ atoms admits a $(k-1)$-dimensional family of mass perturbations that are feasible in both directions, and this forces the suitably scaled Hessian of the value function $w\mapsto\max_{G\in\Pcal(K)}\sum_iw_i\log f_G(X_i)$ to have at least $k-1$ eigenvalues no smaller than one (\Cref{lem:hessian-summary}).  The injective map $w_i\mapsto w_i\{1+\log(\vhat_i(w)/\vhat_i(\one))\}$ turns these eigenvalues into a Jacobian factor exponential in $k$, and comparing the Gamma density before and after the map yields an exponential moment bound on the number of atoms (\Cref{lem:determinant-moment}).  Because the weights are concentrated, the weighted fit stays close to the ordinary one, and the cost of this change of variables is of the order of the effective dimension rather than of $n$.  The Hellinger bound is proved differently, by a likelihood-ratio argument over a finite net of Gaussian mixture densities whose size is controlled by moment matching (\Cref{sec:hellinger}).

Simulations in one and two dimensions (\Cref{sec:numerics}) show that the reweighted NPMLE stays close to the ordinary NPMLE in finite samples.  Along a regularization path, the weight perturbation decays at the expected $\alpha^{-1/2}$ rate, whereas the ordinary log-likelihood gap and the squared Hellinger distance between fitted densities decay approximately as $\alpha^{-1}$.  Across three designs and eleven design--sample-size cells, the reweighted NPMLE with the balanced concentration $\alpha_n=\bar r_n\log n$ has essentially the same Hellinger risk and resolved support size as the ordinary NPMLE.  By contrast, the concentration $\alpha=1$ of the Bayesian bootstrap reduces the resolved support size slightly in some experiments, but at the cost of a large ordinary log-likelihood gap and a larger Hellinger risk.  Concentration of the weights, rather than random reweighting alone, is therefore essential.

The rest of the paper is organized as follows.  \Cref{sec:model} defines the estimator and states the main statistical results.  \Cref{sec:effective} gives the abstract effective-dimension theorem.  \Cref{sec:gaussian} specializes the geometry to Gaussian mixtures, proves uniqueness and bounds the Hellinger loss.  \Cref{sec:numerics} studies finite-sample sparsity and proximity to the ordinary NPMLE.  \Cref{sec:discussion} discusses limitations and extensions.  Proofs and details of the numerical study are in the appendices.

\paragraph{Notation.}
For vectors $u,v\in(0,\infty)^n$, the expressions $\log v$, $u/v$, and $uv$ act coordinatewise; $\one=(1,\ldots,1)^\top$, and $\diag(w)$ is the diagonal matrix with diagonal $w$.  The number of atoms of a discrete distribution $G$ is $\#\supp(G)$, both $\norm{\cdot}$ and $\norm{\cdot}_2$ denote the Euclidean norm, and $B(0,S)$ is the closed Euclidean ball of radius $S$.  For symmetric matrices, $A\succeq B$ means that $A-B$ is positive semidefinite.  For probability measures $P$ and $Q$, $\norm{P-Q}_{\TV}=\sup_A\abs{P(A)-Q(A)}$ is the total variation distance.  For positive sequences, $a_n\asymp b_n$ means that $a_n/b_n$ is bounded away from zero and infinity, and $\log$ is the natural logarithm.  The symbols $C,c,C_b,\ldots$ denote positive constants that do not depend on $n$ or on the true mixing distribution; their values may change from one occurrence to the next, and subscripts indicate additional dependence.

\section{Model, reweighted NPMLE, and main results}\label{sec:model}

This section introduces the model and the estimator and then states the main results.

\subsection{Gaussian location mixtures}\label{sec:mixtures}

Throughout the paper we work under the following standing assumption.

\begin{assumption}\label{ass:standing}
The dimension $d\ge1$ is fixed, and $K\subset B(0,S)\subset\RR^d$ is a fixed nonempty compact set with $S<\infty$.
\end{assumption}

\noindent Under \Cref{ass:standing}, generic constants may depend on $d$, $K$ and $S$ without further mention.  Let $\Pcal(K)$ denote the set of Borel probability measures on $K$, and define
\[
 \phi_d(x)=(2\pi)^{-d/2}\exp(-\norm{x}^2/2),
 \qquad f_G(x)=\int_K\phi_d(x-\theta)\,dG(\theta).
\]
The data are generated according to
\begin{equation}\label{eq:model}
 X_i=\Theta_i+Z_i,\qquad
 \Theta_i\stackrel{\mathrm{iid}}\sim G^*,\quad
 Z_i\stackrel{\mathrm{iid}}\sim N(0,I_d),
\end{equation}
where $G^*$ is a Borel probability measure on $\RR^d$, called the true mixing distribution, and the two sequences are independent.  The model is correctly specified when $G^*\in\Pcal(K)$.  We write
\[
 L_n(G)=\sum_{i=1}^n\log f_G(X_i),
 \qquad
 H^2(p,q)=\int_{\RR^d}(\sqrt p-\sqrt q)^2\,dx
\]
for the log-likelihood and the squared Hellinger distance; note that $H^2$ is defined without a factor $1/2$.  For $G\in\Pcal(K)$, we call $(f_G(X_1),\ldots,f_G(X_n))$ the fitted vector of $G$ and $f_G$ its fitted density.

An ordinary NPMLE is any $\Ghat_0\in\argmax_{G\in\Pcal(K)}L_n(G)$.  It exists because $\Pcal(K)$ is weakly compact and, since the Gaussian kernel is positive and continuous, $L_n$ is finite and weakly continuous on $\Pcal(K)$.  We do not assume deterministic uniqueness of the maximizing mixing distribution for an arbitrary compact parameter set; all ordinary NPMLEs, however, have the same fitted vector \citep{lindsay1983geometry}, which we call the ordinary fitted vector.  Consequently, the ordinary log-likelihood gap and the log-fit discrepancy
\begin{equation}\label{eq:gap-discrepancy}
 \Delta L_n(G)=L_n(\Ghat_0)-L_n(G),
 \qquad
 D_n(G)=\sum_{i=1}^n\bigl\{\log f_G(X_i)-\log f_{\Ghat_0}(X_i)\bigr\}^2,
\end{equation}
depend on $\Ghat_0$ only through the ordinary fitted vector and are well defined for every $G\in\Pcal(K)$; moreover, $\Delta L_n(G)\ge0$.  These two quantities measure how far a mixing distribution is from the ordinary NPMLE in likelihood and in fitted values; \Cref{thm:main}(iii) bounds both for the reweighted NPMLE.

Restricting the mixing distribution to a compact set is natural for Gaussian location mixtures: every maximizer of $L_n$ over all Borel probability measures on $\RR^d$ is supported in the convex hull of $X_1,\ldots,X_n$ \citep{soloffetal2025}, which in one dimension is the range of the data.  Indeed, moving mass from a point outside the convex hull to its Euclidean projection onto the hull brings it closer to every observation and therefore increases every coordinate of the fitted vector; the same argument applies to the weighted likelihood \eqref{eq:rrnpmle} below, with any positive weights.  Hence, for every compact $K$ containing the convex hull of the data, the maximizers over $\Pcal(K)$ are exactly the maximizers over all probability measures on $\RR^d$.  \Cref{ass:standing} requires more, namely that $K$, and with it the radius $S$ on which our constants depend, stay fixed as $n$ grows.  Since the observations spread over a region of radius of order $\sqrt{\log n}$, the unconstrained estimators correspond to the data-dependent choice $K=\conv\{X_1,\ldots,X_n\}$; extending our results to this choice requires tracking the dependence of the constants on $S$ (\Cref{sec:discussion}).

\subsection{The randomly reweighted NPMLE}

We now define the estimator studied in this paper.  Let $\alpha_n>0$ be a concentration parameter, specified in \Cref{sec:main-theorem}, and for $a>0$ let $\mathrm{Gamma}(a,a)$ denote the Gamma law with shape and rate $a$.  Independently of the data, draw
\begin{equation}\label{eq:weights}
 W_1,\ldots,W_n\stackrel{\mathrm{iid}}\sim\mathrm{Gamma}(\alpha_n,\alpha_n),
\end{equation}
and define the reweighted NPMLE by
\begin{equation}\label{eq:rrnpmle}
 \Ghat_W\in\Mcal_X(W):=
 \argmax_{G\in\Pcal(K)}\sum_{i=1}^nW_i\log f_G(X_i).
\end{equation}
As for the ordinary NPMLE, all members of $\Mcal_X(W)$ have the same fitted vector; when $\Mcal_X(W)$ is a singleton, $\Ghat_W$ denotes its unique element.

The weights are easier to interpret after normalization.  Set $W_+=\sum_iW_i$, $\pi_i=W_i/W_+$, $F_n=n^{-1}\sum_i\delta_{X_i}$ and $F_n^{(\alpha_n)}=\sum_i\pi_i\delta_{X_i}$, and write $Fh=\int h\,dF$ for a measure $F$ and a function $h$.  Dividing the objective by $W_+$ does not change its maximizers, so
\[
 \Mcal_X(W)=\argmax_{G\in\Pcal(K)}
 \int\log f_G\,dF_n^{(\alpha_n)}.
\]
Thus the reweighted NPMLE maximizes the log-likelihood with the empirical distribution $F_n$ replaced by the random distribution $F_n^{(\alpha_n)}$.  The next result shows that this perturbation is centered at $F_n$, vanishes as $\alpha_n\to\infty$, and leaves the population criterion unchanged.

\begin{proposition}[Dirichlet representation and unchanged target]
\label{prop:dirichlet}
Conditional on the data, $(\pi_1,\ldots,\pi_n)\sim\mathrm{Dirichlet}(\alpha_n,\ldots,\alpha_n)$ and is independent of $W_+\sim\mathrm{Gamma}(n\alpha_n,\alpha_n)$.  For every real function $h$ on the sample points,
\[
 \EE(F_n^{(\alpha_n)}h\mid X)=F_nh,
 \qquad
 \operatorname{Var}(F_n^{(\alpha_n)}h\mid X)
 =\frac{F_nh^2-(F_nh)^2}{n\alpha_n+1},
\]
and
\[
 \EE\bigl(\norm{F_n^{(\alpha_n)}-F_n}_{\TV}\mid X\bigr)
 \le \frac12\sqrt{\frac{n-1}{n\alpha_n+1}}.
\]
If $X_i\stackrel{\mathrm{iid}}\sim f_{G^*}$ and $\EE_{G^*}\abs{\log f_G(X_1)}<\infty$, then
\[
 \EE_{X,W}\frac1n\sum_iW_i\log f_G(X_i)
 =\EE_{G^*}\log f_G(X_1)
 =\EE_{X,W}\sum_i\pi_i\log f_G(X_i).
\]
In particular, if $G^*\in\Pcal(K)$, then $\EE_{G^*}\log f_{G^*}(X_1)-\EE_{G^*}\log f_G(X_1)=\KL(f_{G^*}\|f_G)$ for every $G\in\Pcal(K)$, where $\KL(p\|q)=\int p\log(p/q)$ is the Kullback--Leibler divergence.
\end{proposition}

\Cref{prop:dirichlet} shows that the reweighting adds no penalty on the number of atoms and leaves the population target unchanged: in expectation, the reweighted criterion equals the ordinary one, and under correct specification both are maximized by $G^*$.  The perturbation of the empirical distribution is at most of order $\alpha_n^{-1/2}$ in total variation, uniformly in $n$, and therefore vanishes as $\alpha_n\to\infty$.  This is where the reweighted NPMLE departs from the Bayesian bootstrap, in which $\alpha=1$, the perturbation does not vanish, and the maximizers from repeated draws are treated as a distribution.  Here a single draw defines the estimator.

\subsection{The joint sparsity and accuracy theorem}\label{sec:main-theorem}

The main result requires the concentration to grow with $n$.  We state it for the balanced concentration
\begin{equation}\label{eq:scales}
 r_n=\left(\frac{\log n}{\log\log n}\right)^d,
 \qquad \bar r_n=r_n+\log n,
 \qquad \alpha_n=\bar r_n\log n,
\end{equation}
defined for all sufficiently large $n$.  The first term, $r_n$, is of the order of the effective dimension of the set of relative fitted vectors in the Gaussian model (\Cref{prop:gaussian-width}), and the additional $\log n$ in $\bar r_n$ pays for a polynomially small failure probability; it is absorbed by $r_n$ when $d\ge2$.  The sparsity argument of \Cref{sec:effective} requires a concentration of at least a constant multiple of $\bar r_n$, and the extra factor $\log n$ in \eqref{eq:scales} makes the ordinary log-likelihood gap at most of order $1/\log n$.

\begin{theorem}[Uniqueness, sparsity and accuracy]
\label{thm:main}
Suppose that \Cref{ass:standing} holds and that the data follow \eqref{eq:model} with $G^*\in\Pcal(K)$.  Let the weights satisfy \eqref{eq:weights} with the concentration in \eqref{eq:scales}.  For every $b>0$, there are constants $C_b<\infty$ and $n_0$ depending only on $d,K,S,b$ such that, uniformly over $G^*\in\Pcal(K)$, the following event has probability at least $1-C_bn^{-b}$ for every $n\ge n_0$:
\begin{enumerate}[label=\textup{(\roman*)}]
\item The reweighted NPMLE is unique and
\[
  \#\supp(\Ghat_W)\le C_b\bar r_n.
\]
\item The weight perturbation satisfies
\[
  \max_{1\le i\le n}\abs{W_i-1}\le \frac{C_b}{\sqrt{\bar r_n}}.
\]
\item The ordinary log-likelihood gap and the log-fit discrepancy in \eqref{eq:gap-discrepancy} satisfy
\[
  \Delta L_n(\Ghat_W)\le\frac{C_b}{\log n},
  \qquad
  D_n(\Ghat_W)\le\frac{C_b}{\log n}.
\]
\item The squared Hellinger loss satisfies
\[
 H^2(f_{\Ghat_W},f_{G^*})
 \le C_b\frac{(\log n)^{d+1}}
 {n(\log\log n)^d}.
\]
\end{enumerate}
For $d\ge2$, the support bound in part \textup{(i)} is $O\{(\log n/\log\log n)^d\}$; for $d=1$, it is $O(\log n)$.
\end{theorem}

Parts (i)--(iii) of \Cref{thm:main} are structural.  Their proofs use the data only through the radius condition $\max_i\norm{X_i}\le C_0\sqrt{\log n}$ and, for uniqueness, the absolute continuity of the distribution of $(X_1,\ldots,X_n)$; in particular, they do not require $G^*\in\Pcal(K)$ (\Cref{cor:conditional} and \Cref{prop:generic}).  Only part (iv) requires correct specification, because its proof compares the log-likelihood of $\Ghat_W$ with that of $G^*$.

The balanced concentration is convenient but not necessary: much smaller weight perturbations, that is, much larger concentrations, also produce exact sparsity.

\begin{corollary}[Arbitrarily small weight perturbations]
\label{cor:tiny}
Under the assumptions of \Cref{thm:main}, fix any $L>0$ and instead take $\alpha_n=n^{2L+2}$.  For every $b>0$, with probability at least $1-C_bn^{-b}$ for all sufficiently large $n$, the reweighted NPMLE is unique,
\[
 \#\supp(\Ghat_W)\le C_{b,L}\bar r_n,
 \qquad
 \max_i\abs{W_i-1}\le n^{-L},
\]
and the Hellinger conclusion of \Cref{thm:main}(iv) holds.
\end{corollary}

Thus an inverse-polynomially small weight perturbation already yields a unique maximizer with polylogarithmically many atoms, whose fitted density retains the Hellinger rate of \Cref{thm:main}.  Because $K$ is compact and Gaussian location mixtures are identifiable, Hellinger convergence of the fitted density also implies convergence of the mixing distribution.

\begin{corollary}[Weak consistency of the mixing distribution]
\label{cor:weak}
Under the assumptions of \Cref{thm:main}, $\Ghat_W$ converges weakly to $G^*$ in probability: $d(\Ghat_W,G^*)\to0$ in probability for every metric $d$ that metrizes weak convergence on $\Pcal(K)$.
\end{corollary}

The proof of \Cref{thm:main} combines results of independent interest.  For a general positive kernel, \Cref{thm:effective} shows that, with high probability, the number of atoms of every extreme optimizer is bounded by a constant multiple of the effective dimension of the set of relative fitted vectors plus a logarithmic term, and that the ordinary log-likelihood gap and the log-fit discrepancy are small.  For Gaussian kernels, \Cref{prop:gaussian-width} shows that this effective dimension is $O(r_n)$; with a Gamma tail bound, this yields \Cref{cor:conditional}, which contains parts (ii) and (iii) and bounds the number of atoms of every extreme optimizer.  \Cref{prop:generic} then turns this bound into uniqueness, which gives part (i).  Part (iv) follows from \Cref{thm:near-mle}, which bounds the Hellinger loss of every mixing distribution with log-likelihood at least $L_n(G^*)-1$, because part (iii) gives $L_n(\Ghat_W)\ge L_n(\Ghat_0)-1\ge L_n(G^*)-1$ for large $n$.

\section{An effective-dimension bound on the number of atoms}
\label{sec:effective}

This section develops the sparsity part of \Cref{thm:main} for an arbitrary positive continuous kernel.  The argument uses a single property of the fitted vectors: relative to the ordinary fitted vector, they lie close to a low-dimensional linear subspace.  Under this condition, \Cref{thm:effective} bounds the number of atoms of every extreme optimizer, with high probability, by a constant multiple of the dimension of that subspace plus a logarithmic term; its proof, sketched in \Cref{sec:hessian}, turns the curvature of the value function into this bound through a change of variables for the weights.  \Cref{sec:gaussian} verifies the condition for Gaussian kernels and shows that, for them, sparsity of the extreme optimizers implies uniqueness.

\subsection{Positive kernel hulls and optimizer fibers}\label{sec:effective-setup}

The weighted log-likelihood depends on a mixing distribution only through its fitted vector, so we formulate the problem directly in terms of fitted vectors, for a general positive kernel and a fixed dataset.  Let $\Theta$ be a nonempty compact metric space and let $A:\Theta\to(0,\infty)^n$ be a continuous map, which we call the evaluation map; its values $A(\theta)$ are the evaluation vectors.  In the Gaussian model with data $x_1,\ldots,x_n$, $\Theta=K$ and $A(\theta)=(\phi_d(x_i-\theta))_{i\le n}$.  Write
\[
 \Acal=A(\Theta),\qquad \Ccal=\conv(\Acal).
\]
The set $\Ccal$ of fitted vectors consists exactly of the vectors $\int A(\theta)\,dG(\theta)$ with $G\in\Pcal(\Theta)$.  For $w\in(0,\infty)^n$, define
\begin{equation}\label{eq:value}
 \Psi(w)=\max_{v\in\Ccal}w^\top\log v,
 \qquad
 \vhat(w)=\argmax_{v\in\Ccal}w^\top\log v,
 \qquad z(w)=\log\vhat(w).
\end{equation}
We call $\Psi$ the value function and $\vhat(w)$ the weighted fit.  Because $\Ccal$ is compact and $v\mapsto w^\top\log v$ is strictly concave, the weighted fit exists and is unique (\Cref{lem:regularity}), and the set of maximizing mixing distributions is the \emph{optimizer fiber}
\begin{equation}\label{eq:fiber}
 \Mcal_A(w)=\left\{G\in\Pcal(\Theta):
       \int A(\theta)\,dG(\theta)=\vhat(w)\right\}.
\end{equation}
In the Gaussian model, $\Mcal_A(W)$ is the set $\Mcal_X(W)$ of reweighted NPMLEs in \eqref{eq:rrnpmle}.

Recall that a member of a convex set is \emph{extreme} if it is not a nontrivial convex combination of two other members.  We call the extreme points of $\Mcal_A(w)$ \emph{extreme optimizers}.  Because $\Mcal_A(w)$ is compact and convex, the Krein--Milman theorem represents it as the closed convex hull of its extreme optimizers, so we first bound the largest number of atoms among them,
\begin{equation}\label{eq:sext}
 \sext(w)=
 \max\left\{\#\supp(G):G\text{ is an extreme point of }\Mcal_A(w)\right\}.
\end{equation}
This maximum is well defined and lies in $\{1,\ldots,n\}$, because an optimizer is extreme exactly when it is finitely supported and the evaluation vectors of its atoms are linearly independent (\Cref{lem:extreme}).  For Gaussian kernels and almost every dataset, \Cref{prop:generic} turns a bound $\sext(w)\le m$ with $2m(d+1)\le n$ into uniqueness of the optimizer.

\subsection{Effective dimension}\label{sec:effective-dimension}

Sparsity will be controlled by how close the fitted vectors are to a low-dimensional subspace, measured relative to the ordinary fitted vector $v_0=\vhat(\one)$.  Relative fitted vectors are natural here because the weighted objective depends on $v$ through $\log v$, and $\log(v/v_0)\approx v/v_0-\one$ when $v$ is close to $v_0$.  Define the set of relative fitted vectors
\begin{equation}\label{eq:relative-set}
 \Ucal=\{v/v_0-\one:v\in\Ccal\}.
\end{equation}

\begin{definition}[Effective dimension and residual width]
\label{def:effective}
A rank-$r$ orthogonal projection $P$ on $\RR^n$ has \emph{residual width} $\eta$ if
\begin{equation}\label{eq:width}
 \sup_{u\in\Ucal}\norm{(I-P)u}_2\le\eta.
\end{equation}
We then call $r$ an \emph{effective dimension} of $\Ucal$.  Both $r$ and $\eta$ may depend on the fixed data and the evaluation map, but not on the random weights.
\end{definition}

We also need a notion of approximate optimality, which covers numerically computed solutions (\Cref{cor:approx}).  For $\tau\ge0$, a vector $\vtil\in\Ccal$ is a \emph{$\tau$-optimal weighted fit} if
\begin{equation}\label{eq:tau-opt}
 w^\top\log\vtil\ge\Psi(w)-\tau,
\end{equation}
so that the tolerance is measured in units of total weighted log-likelihood.

\begin{theorem}[Effective-dimension bound on the number of atoms]
\label{thm:effective}
There are universal constants $C,c>0$ such that the following holds.  Let $P$ be a rank-$r$ orthogonal projection with residual width $\eta$, as in \Cref{def:effective}.  Fix $q\ge1$, put $Q=q+\log(2n)$, and suppose
\begin{equation}\label{eq:effective-conditions}
 \alpha\ge C(r+Q),
 \qquad
 \eta\sqrt{\alpha nQ}\le1.
\end{equation}
Let $W_i\stackrel{\mathrm{iid}}\sim\mathrm{Gamma}(\alpha,\alpha)$.  With probability at least $1-3e^{-q}$,
\begin{equation}\label{eq:effective-support}
 \sext(W)\le C(r+q),
\end{equation}
\begin{equation}\label{eq:effective-exact-fit}
 0\le \one^\top\log v_0-\one^\top\log\vhat(W)
 \le\frac{C(r+q)}{\alpha},
 \qquad
 \norm{\log\{\vhat(W)/v_0\}}_2^2
 \le\frac{C(r+q)}{\alpha}.
\end{equation}
On the same event, simultaneously for every $0\le\tau<c$ and every $\tau$-optimal fit $\vtil$,
\begin{equation}\label{eq:effective-approx-fit}
 0\le \one^\top\log v_0-\one^\top\log\vtil
 \le C\left\{\frac{r+q}{\alpha}+\tau\right\},
 \qquad
 \norm{\log(\vtil/v_0)}_2^2
 \le C\left\{\frac{r+q}{\alpha}+\tau\right\}.
\end{equation}
\end{theorem}

\Cref{thm:effective} shows how concentrated the weights must be: by \eqref{eq:effective-conditions}, the concentration need only exceed a constant multiple of the effective dimension plus a logarithmic term.  Larger concentrations bring the weighted objective closer to the ordinary one but require a smaller residual width, and for Gaussian kernels analyticity allows both conditions to hold even when $\alpha$ is polynomially large (\Cref{prop:gaussian-width} and \Cref{cor:tiny}).  The bound \eqref{eq:effective-support} concerns exact extreme optimizers: for a mixing distribution whose fitted vector is only $\tau$-optimal, \eqref{eq:effective-approx-fit} controls the log-likelihood gap and the log-fit discrepancy but not the number of atoms.

\subsection{Curvature and change of variables}\label{sec:hessian}

We now explain where the bound \eqref{eq:effective-support} of \Cref{thm:effective} comes from.  By \Cref{lem:regularity}, the value function $\Psi$ is convex and continuously differentiable with $\nabla\Psi=z$, and $z$ is locally Lipschitz, so by Rademacher's theorem the Hessian $\nabla^2\Psi$ exists almost everywhere.  Now suppose that an extreme optimizer has $k$ atoms.  Because the evaluation vectors of its atoms are linearly independent (\Cref{lem:extreme}), moving mass among the atoms while keeping the total mass fixed produces a $(k-1)$-dimensional space of perturbations of the fitted vector, each feasible in both directions.  A second-order comparison of the weighted objective along these perturbations yields the following lower bound on the curvature of $\Psi$.

\begin{lemma}[Support--sensitivity inequality]
\label{lem:hessian-summary}
The Hessian $\nabla^2\Psi(w)$ exists at almost every $w\in(0,\infty)^n$, and the exceptional set depends only on the evaluation map and not on a selected optimizer.  At each point where it exists and for every extreme optimizer $G\in\Mcal_A(w)$ with $k$ atoms, there is an orthogonal projection $P_G$ of rank $k-1$ such that $\diag(w)^{1/2}\nabla^2\Psi(w)\diag(w)^{1/2}\succeq P_G$.  In particular, for almost every $w$ there is an orthogonal projection $P_*(w)$ of rank $\sext(w)-1$ such that
\begin{equation}\label{eq:hessian-summary}
 \diag(w)^{1/2}\nabla^2\Psi(w)\diag(w)^{1/2}\succeq P_*(w).
\end{equation}
\end{lemma}

The matrix inequality \eqref{eq:hessian-summary} is stronger than a bound on the trace of the Hessian: it yields a determinant factor that is exponential in the number of atoms.  To exploit it, write $t(w)=z(w)-z(\one)$ for the logarithmic ratio of the weighted and ordinary fitted vectors, let $Z(w)=(w-\one)^\top t(w)$, and consider the locally Lipschitz map $\mathcal F$ with coordinates
\begin{equation}\label{eq:change-map}
 \mathcal F_i(w)=w_i\{1+t_i(w)\}.
\end{equation}
On a localization event $E$, on which the weights are close to one and hence the weighted fit is close to $v_0$, the map $\mathcal F$ is injective and, by \eqref{eq:hessian-summary}, its Jacobian determinant is at least $\prod_i\{1+t_i(w)\}(17/9)^{\sext(w)-1}$ at almost every point.  Comparing the Gamma density at $\mathcal F(w)$ with that at $w$ and applying the injective area formula \citep{evansgariepy2015} then gives the moment bound of \Cref{lem:determinant-moment},
\begin{equation}\label{eq:moment-summary}
 \EE\left[\bm 1_{\{W\in E\}}
 \left(\frac{17}{9}\right)^{\sext(W)-1}
 e^{-\alpha\{Z(W)+\norm{t(W)}_2^2\}}\right]\le1.
\end{equation}
Together with the concentration of the projected weight perturbation $P(W-\one)$, the residual-width bound \eqref{eq:width} bounds the exponent $\alpha\{Z(W)+\norm{t(W)}_2^2\}$ in \eqref{eq:moment-summary} by $O(r+q)$ with high probability, and Markov's inequality then yields \eqref{eq:effective-support}.  The change of variables \eqref{eq:change-map} is what makes the sparsity exact: the moment bound \eqref{eq:moment-summary} controls the number of atoms itself, rather than the size of small masses or an approximation error.

\section{Gaussian mixtures: effective dimension, uniqueness and accuracy}
\label{sec:gaussian}

This section specializes \Cref{thm:effective} to Gaussian mixtures and supplies the remaining ingredients of the proof of \Cref{thm:main}.  For a dataset $x=(x_1,\ldots,x_n)$ with radius $T=\max_i\norm{x_i}$, the Gaussian evaluation map is
\[
 A(\theta)=\bigl(\phi_d(x_1-\theta),\ldots,
                    \phi_d(x_n-\theta)\bigr),
 \qquad \theta\in K.
\]
\Cref{sec:gaussian-width} shows that this map satisfies the hypothesis of \Cref{thm:effective} with effective dimension $O(r_n)$, and \Cref{sec:uniqueness} shows that the resulting sparsity of the extreme optimizers implies uniqueness; both subsections treat the dataset as fixed.  \Cref{sec:hellinger} returns to the sampling model \eqref{eq:model} and bounds the Hellinger loss.

\subsection{Analytic low-rank approximation}\label{sec:gaussian-width}

The identity
\[
 \phi_d(x_i-\theta)=\phi_d(x_i)
 \exp\{x_i^\top\theta-\norm{\theta}^2/2\}
\]
reduces the geometry of Gaussian fitted vectors to that of multivariate exponential functions, and truncating their Taylor series yields the required effective dimension.

\begin{proposition}[Gaussian residual-width bound]
\label{prop:gaussian-width}
Suppose that \Cref{ass:standing} holds and let $G_0\in\Pcal(K)$ be arbitrary.  Put $v_{0,i}=f_{G_0}(x_i)$.  For every integer $m\ge0$, there is an orthogonal projection $P_m$ of rank at most
\begin{equation}\label{eq:rank-m}
 1+\binom{m+d}{d}
\end{equation}
whose range contains $\one$ and such that
\begin{equation}\label{eq:eta-m}
 \sup_{G\in\Pcal(K)}
 \norm{(I-P_m)\{f_G(x)/v_0-\one\}}_2
 \le \eta_m:=\sqrt n\,\eps_m,
\end{equation}
where
\begin{equation}\label{eq:eps-m}
 \eps_m=\exp(2ST+S^2/2)\frac{(ST)^{m+1}}{(m+1)!}.
\end{equation}
When $ST=0$, one may take $\eps_m=0$.
\end{proposition}

If $T\le C_0\sqrt{\log n}$ and $m\asymp\log n/\log\log n$, the rank bound \eqref{eq:rank-m} is $O(r_n)$, while $\eps_m$ falls below any prescribed inverse power of $n$ once the constant in $m$ is large enough.  The second condition in \eqref{eq:effective-conditions} therefore holds for every concentration that is at most polynomial in $n$.  With the balanced concentration, \Cref{prop:gaussian-width} and \Cref{thm:effective} give the following bounds, conditional on the data.

\begin{corollary}[Conditional structural and likelihood bounds]
\label{cor:conditional}
Suppose that \Cref{ass:standing} holds.  Fix $C_0<\infty$ and a dataset with $T\le C_0\sqrt{\log n}$.  Let $r_n,\bar r_n$ be as in \eqref{eq:scales}, set $\alpha_n=\bar r_n\log n$, and draw independent $\mathrm{Gamma}(\alpha_n,\alpha_n)$ weights.  For every $b>0$, with conditional probability at least $1-C_bn^{-b}$,
\[
 \sext(W)\le C_b\bar r_n,
 \qquad
 \max_i\abs{W_i-1}\le C_b\bar r_n^{-1/2},
\]
and every reweighted NPMLE satisfies the two bounds in \Cref{thm:main}(iii).
\end{corollary}

\Cref{cor:conditional} holds for every fixed dataset of radius at most $C_0\sqrt{\log n}$; it involves no sampling model and no correct-specification assumption.

\subsection{From sparse extreme optimizers to a unique maximizer}\label{sec:uniqueness}

A bound on the number of atoms of the extreme optimizers does not by itself make every optimizer sparse, since a nontrivial convex combination of extreme optimizers can have more atoms, and a limit of such combinations can have infinite support.  For Gaussian evaluations and generic data, this cannot happen once all extreme optimizers are sufficiently sparse.

\begin{proposition}[Generic evaluation independence]
\label{prop:generic}
For each $n,d$, there is a Borel set $\mathcal X_{n,d}\subset(\RR^d)^n$ whose complement has Lebesgue measure zero such that the following holds simultaneously for every $x\in\mathcal X_{n,d}$.  If $k(d+1)\le n$ and $\theta_1,\ldots,\theta_k\in\RR^d$ are distinct, then the $n\times k$ matrix
\begin{equation}\label{eq:evaluation-matrix}
 \bigl(\phi_d(x_i-\theta_j)\bigr)_{i\le n,j\le k}
\end{equation}
has linearly independent columns.  The locations may depend on the data.  Consequently, if every extreme point of a Gaussian optimizer fiber has at most $m$ atoms and $2m(d+1)\le n$, then that fiber is a singleton and its unique member has at most $m$ atoms.
\end{proposition}

\Cref{prop:generic} holds simultaneously for all of the uncountably many possible locations, including locations chosen after seeing the data; this is what allows it to be applied to the atoms of an optimizer, which depend on the data.  Its proof views a linear dependence as a point of an analytic incidence set with $k(d+1)-1$ free parameters.  Since this is fewer than the $n$ equations imposed by the sample, the incidence set projects onto a Lebesgue-null set of datasets.  Uniqueness then follows by comparing two extreme optimizers and invoking the Krein--Milman theorem.

Uniqueness matters because several quantities of interest depend on the mixing distribution beyond its fitted vector.  Examples are the fitted density $f_{\Ghat}(x)$ at points $x$ other than the sample points and the empirical Bayes posterior mean $\int\theta\,\phi_d(x-\theta)\,d\Ghat(\theta)/f_{\Ghat}(x)$, which is the basis of NPMLE-based denoising \citep{koenkermizera2014,sahaguntuboyina2020,soloffetal2025}.  When the maximizer is not unique, different maximizers, and hence different algorithms, can return different values of these quantities.  For the unrestricted problem, the ordinary NPMLE is unique for every dataset in one dimension \citep{lindsay1983partii}.  Recent work establishes uniqueness for Lebesgue-almost every dataset in every dimension \citep{wang2026uniqueness}.

Reweighting is what makes \Cref{prop:generic} applicable.  The proposition requires every extreme optimizer to have at most $m$ atoms with $2m(d+1)\le n$, whereas in several dimensions the classical general bound on the support size of an extreme ordinary optimizer is $n$ atoms.  For the reweighted NPMLE, \Cref{cor:conditional} gives $m=C_b\bar r_n$, which satisfies $2m(d+1)\le n$ for all sufficiently large $n$.  The random weights therefore supply the sparsity that this proof of uniqueness via \Cref{prop:generic} requires; the contribution is the exact polylogarithmic support bound, rather than almost-sure uniqueness alone.

\subsection{Hellinger accuracy}\label{sec:hellinger}

We now return to the sampling model \eqref{eq:model} with $G^*\in\Pcal(K)$ and convert the bound $\Delta L_n(\Ghat_W)\le C_b/\log n$ of \Cref{cor:conditional} into a bound on the Hellinger loss; this is the only step in the proof of \Cref{thm:main} that uses correct specification.  Since $L_n(\Ghat_0)\ge L_n(G^*)$, every $G$ with $\Delta L_n(G)\le1$ has log-likelihood at least $L_n(G^*)-1$.  \Cref{thm:near-mle} bounds the Hellinger loss of every mixing distribution with log-likelihood at least $L_n(G^*)-1$, simultaneously and on a single event that depends only on the data, so it applies to the reweighted NPMLE whatever the value of the weights.  Bounds of this type were proved for the ordinary NPMLE by \citet{zhang2009} in one dimension and by \citet{sahaguntuboyina2020} in several dimensions, and they also apply to approximate maximizers of the likelihood.

\begin{proposition}[Uniform Hellinger bound]
\label{thm:near-mle}
Suppose that \Cref{ass:standing} holds and that the data follow \eqref{eq:model} with $G^*\in\Pcal(K)$.  For every $b>0$, there are constants $C_b<\infty$ and $n_0$ depending only on $d,K,S,b$ such that, uniformly over $G^*\in\Pcal(K)$, the following event has probability at least $1-C_bn^{-b}$ for every $n\ge n_0$: simultaneously for all $G\in\Pcal(K)$,
\begin{equation}\label{eq:near-truth}
 L_n(G)\ge L_n(G^*)-1
 \quad\Longrightarrow\quad
 H^2(f_G,f_{G^*})
 \le C_b\frac{(\log n)^{d+1}}
 {n(\log\log n)^d}.
\end{equation}
\end{proposition}

The proof of \Cref{thm:near-mle} follows the classical route from metric entropy to Hellinger rates of maximum likelihood estimators \citep{wongshen1995,ghosalvandervaart2001}.  \Cref{lem:finite-net} constructs a deterministic finite set $\Ncal_n$ of Gaussian mixture densities that approximates every $f_G$, $G\in\Pcal(K)$, to within $n^{-3}$, both in Hellinger distance and in log density on a ball of radius $O(\sqrt{\log n})$ that contains all observations with high probability, and whose size satisfies
\begin{equation}\label{eq:net-entropy-summary}
 \log\abs{\Ncal_n}
 \le C\frac{(\log n)^{d+1}}{(\log\log n)^d}.
\end{equation}
Since $\EE_{G^*}\prod_i\{p(X_i)/f_{G^*}(X_i)\}^{1/2}=\{1-H^2(p,f_{G^*})/2\}^n$ for every fixed density $p$, Markov's inequality and a union bound over $\Ncal_n$ show that, with high probability, every $p\in\Ncal_n$ whose log-likelihood ratio $\sum_i\log\{p(X_i)/f_{G^*}(X_i)\}$ is at least $-2$ satisfies $H^2(p,f_{G^*})=O\{(\log\abs{\Ncal_n}+\log n)/n\}$.  If all observations lie in the ball, the member of $\Ncal_n$ that approximates a mixing distribution $G$ with $L_n(G)\ge L_n(G^*)-1$ has log-likelihood ratio at least $-2$, and since $f_G$ lies within Hellinger distance $n^{-3}$ of it, \eqref{eq:net-entropy-summary} yields \eqref{eq:near-truth}.

For mixing distributions supported in a fixed compact set, \citet{sahaguntuboyina2020} obtain the rate $(\log n)^{d+1}/n$.  The bound \eqref{eq:near-truth}, which applies equally to the ordinary NPMLE, is smaller by the factor $(\log\log n)^d$ because the net is smaller.  Each member of $\Ncal_n$ is the density of a discrete mixing distribution with at most $\binom{2m+d}{d}$ atoms, obtained by matching the moments of some $G\in\Pcal(K)$ up to total degree $2m$ and rounding its atoms and masses to a fine grid.  For $\theta\in K$ and $\norm{x}=O(\sqrt{\log n})$, the exponent in $\phi_d(x-\theta)=\phi_d(x)\exp(x^\top\theta-\norm{\theta}^2/2)$ is $O(\sqrt{\log n})$, so $m\asymp\log n/\log\log n$ already gives a polynomially small relative error, as in the one-dimensional approximation results of \citet{wuyang2020}.  In exchange, \Cref{thm:near-mle} requires $G^*$ to be supported in the known set $K$, whereas the bounds of \citet{zhang2009} and \citet{sahaguntuboyina2020} also allow mixing distributions with unbounded support.

Sharper logarithmic factors are available in one dimension.  For mixing distributions supported in a fixed interval, \citet[Appendix~E, Theorems~18--19]{chenwu2026} show that the minimax expected squared Hellinger risk is of order $(\log n)/(n\log\log n)$ and that the support-constrained ordinary NPMLE attains $(\log n)/n$.  \Cref{thm:near-mle} instead gives a high-probability bound in every dimension, uniformly over mixing distributions with log-likelihood at least $L_n(G^*)-1$; its role here is to turn the likelihood-gap bound of \Cref{cor:conditional} into an accuracy guarantee, and we do not claim optimal logarithmic factors.

Numerical algorithms return only approximate maximizers of the reweighted likelihood.  \Cref{cor:approx} shows that these retain the likelihood and Hellinger guarantees of \Cref{thm:main}; its proof combines the approximate-fit bounds \eqref{eq:effective-approx-fit} of \Cref{thm:effective} with \Cref{thm:near-mle}.

\begin{corollary}[Approximate maximizers]
\label{cor:approx}
Under the assumptions of \Cref{thm:main}, let $\tau_n\le c_0/\log n$, where $c_0$ is a fixed sufficiently small constant.  With probability at least $1-C_bn^{-b}$, every $\Gtil\in\Pcal(K)$ whose fitted vector is a $\tau_n$-optimal weighted fit \eqref{eq:tau-opt} for the weights $W$, that is,
\[
 \sum_{i=1}^nW_i\log f_{\Gtil}(X_i)
 \ge\max_{G\in\Pcal(K)}\sum_{i=1}^nW_i\log f_G(X_i)-\tau_n,
\]
satisfies
\[
 \Delta L_n(\Gtil)\le\frac{C_b}{\log n},
 \qquad
 D_n(\Gtil)\le\frac{C_b}{\log n},
 \qquad
 H^2(f_{\Gtil},f_{G^*})\le C_b\frac{(\log n)^{d+1}}{n(\log\log n)^d}.
\]
\end{corollary}

\Cref{cor:approx} controls the likelihood, the fitted vector and the fitted density of an approximate maximizer, but not its number of atoms: adding atoms of negligible mass changes the likelihood arbitrarily little, so the support bound of \Cref{thm:main}(i) applies only to the exact reweighted NPMLE.

\section{Numerical evidence: sparsity and proximity to the ordinary NPMLE}
\label{sec:numerics}

We compare the reweighted NPMLE with the ordinary NPMLE on simulated data, with three questions in mind: whether its empirical density accuracy and resolved support size are comparable to those of the ordinary NPMLE; whether its distance from the ordinary NPMLE decreases at the rates suggested by \Cref{thm:effective} as the concentration grows; and whether these properties require concentrated weights.  The experiments are illustrative and restricted to one and two dimensions: resolving the support of an NPMLE requires an accurate search over the parameter set, whose cost grows rapidly with the dimension, and scalable computation of the NPMLE in higher dimensions is a separate problem \citep{zhangetal2024}.  The designs, the computation and the accuracy checks are described in \Cref{app:numerics}.

We consider three correctly specified designs: a three-point and a uniform mixing distribution in one dimension, and a four-point mixing distribution in two dimensions, with sample sizes between 200 and 2000 (\Cref{app:numerics-design}).  We compare the ordinary NPMLE with the reweighted NPMLE at three concentrations, which label the fits in the figures and in \Cref{tab:numerical-summary}: \emph{Balanced}, the balanced concentration \eqref{eq:scales}; \emph{Tiny}, $\alpha_n=n^3$, which corresponds to \Cref{cor:tiny} with $L=1/2$; and \emph{BB-scale}, $\alpha=1$, the concentration of the Bayesian bootstrap.  The first two satisfy the hypotheses of \Cref{thm:main} and \Cref{cor:tiny}; the third does not, because \Cref{thm:effective} requires the concentration to exceed a multiple of the effective dimension, and it serves as a control.  For each fit $\widehat G$ we report the resolved support size, that is, the number of clusters of atoms at a fixed spatial resolution, together with the squared Hellinger loss $H^2(f_{\widehat G},f_{G^*})$, the squared Hellinger distance $H^2(f_{\widehat G},f_{\widehat G_0})$ to the ordinary NPMLE, and the ordinary log-likelihood gap $\Delta L_n(\widehat G)$ and log-fit discrepancy $D_n(\widehat G)$ of \eqref{eq:gap-discrepancy}.

We first fix one sample of size $n=500$ from the three-point design and vary $\alpha$ from 1 to $n^3$, drawing 30 weight vectors at each value, so that only the concentration changes (\Cref{fig:path}).

\begin{figure}[htb]
\centering
\begin{minipage}[t]{0.36\textwidth}
  \centering\includegraphics[width=\linewidth]{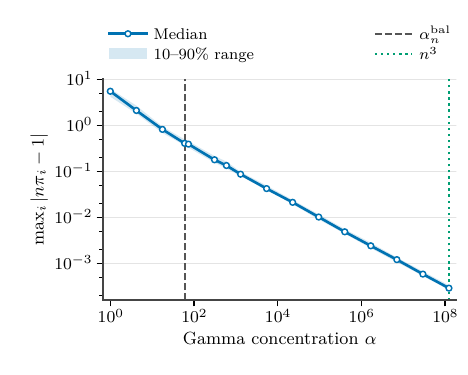}
  \par\vspace{-0.5em}\centerline{\small (a) Weight perturbation}
\end{minipage}\hspace{0.04\textwidth}
\begin{minipage}[t]{0.36\textwidth}
  \centering\includegraphics[width=\linewidth]{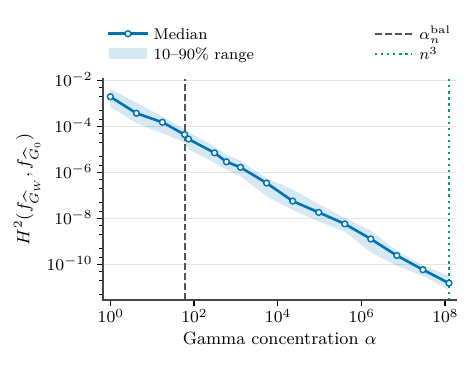}
  \par\vspace{-0.5em}\centerline{\small (b) Distance to the ordinary NPMLE}
\end{minipage}

\par\smallskip
\begin{minipage}[t]{0.36\textwidth}
  \centering\includegraphics[width=\linewidth]{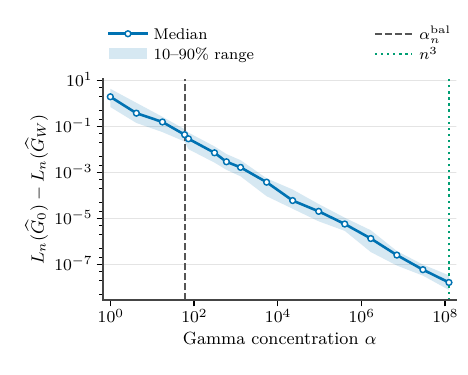}
  \par\vspace{-0.5em}\centerline{\small (c) Ordinary log-likelihood gap}
\end{minipage}\hspace{0.04\textwidth}
\begin{minipage}[t]{0.36\textwidth}
  \centering\includegraphics[width=\linewidth]{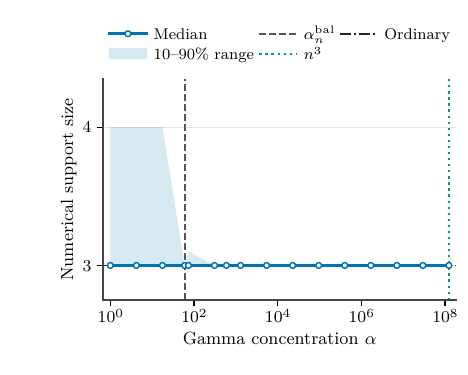}
  \par\vspace{-0.5em}\centerline{\small (d) Resolved support size}
\end{minipage}
\caption{Regularization path on a fixed $n=500$ sample from the three-point mixture.  Curves show medians across 30 independent Gamma-weight draws and bands show the 10th and 90th percentiles.  The dashed and dotted vertical lines mark the balanced concentration and $n^3$, respectively.}
\label{fig:path}
\end{figure}

The path exhibits the two rates suggested by the theory.  A log--log regression of the medians for $\alpha\ge10$ gives slope $\numSlopeWeight$ for the weight perturbation $\max_i|\widetilde W_i-1|$ of the normalized weights $\widetilde W_i=nW_i/\sum_jW_j$, matching the $\alpha^{-1/2}$ scale of Gamma fluctuations (\Cref{lem:gamma-concentration}).  The slopes are $\numSlopeLoss$ for the ordinary log-likelihood gap and $\numSlopeLogfit$ for the log-fit discrepancy, in line with the $1/\alpha$ scaling of the bounds \eqref{eq:effective-exact-fit}, and $\numSlopeDistance$ for the squared Hellinger distance to the ordinary NPMLE.  The fitted vector thus moves at first order in the weight perturbation, whereas the three discrepancies are second order.

At the balanced concentration, the median squared Hellinger distance to the ordinary NPMLE is $\numPathBalancedDistance$ and the median ordinary log-likelihood gap is $\numPathBalancedLoss$, while the median squared Hellinger loss, $\numPathBalancedRisk$, is essentially that of the ordinary NPMLE, $\numPathOrdinaryRisk$.  At $\alpha=n^3$, both discrepancies fall below the numerical accuracy of the fits (\Cref{app:numerics-accuracy}).  The median resolved support size is $\numPathBalancedMedianSupport$ at the balanced concentration and $\numPathTinyMedianSupport$ at $\alpha=n^3$, whereas $\alpha=1$ gives a median distance of $\numPathBBDistance$ and a median gap of $\numPathBBLoss$; \Cref{app:numerics-results} reports the full distributions.

\begin{figure}[htb]
\centering
\begin{minipage}[t]{0.36\textwidth}
  \centering\includegraphics[width=\linewidth]{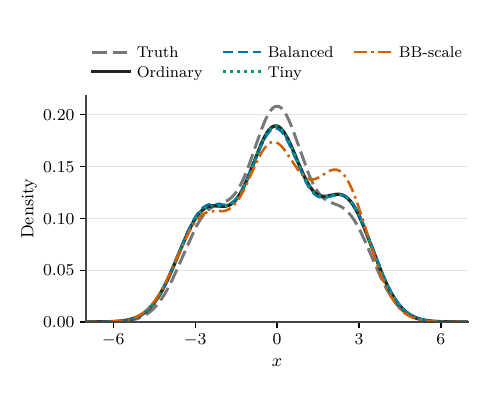}
  \par\vspace{-0.5em}\centerline{\small (a) Fitted densities}
\end{minipage}\hspace{0.04\textwidth}
\begin{minipage}[t]{0.36\textwidth}
  \centering\includegraphics[width=\linewidth]{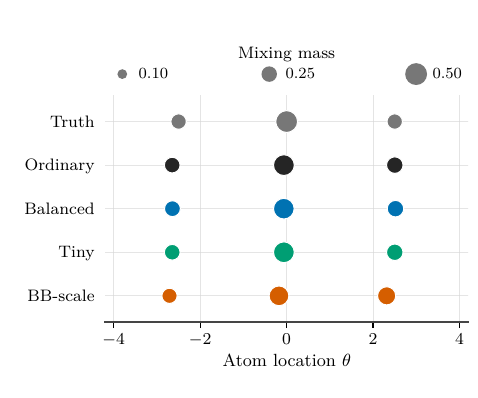}
  \par\vspace{-0.5em}\centerline{\small (b) Mixing distributions}
\end{minipage}
\caption{Representative fits from the regularization-path experiment, one draw for each concentration (\Cref{app:numerics-measures}).  Marker areas in panel (b) are proportional to atom mass.}
\label{fig:representative}
\end{figure}

\Cref{fig:representative} shows why proximity is best judged through the fitted density, the fitted vector and the likelihood.  The fitted densities of the ordinary NPMLE and of the Balanced and Tiny fits overlap closely, whereas the BB-scale fit is visibly displaced; yet at the balanced concentration the median Wasserstein distance between the mixing distributions, $\numPathBalancedWasserstein$, is much larger than the Hellinger distance between the fitted densities.  Gaussian deconvolution is ill posed, so a small perturbation of the density need not produce comparable stability of the atoms; accordingly, \Cref{thm:main} controls the fitted vector, the likelihood and the density rather than the atom locations.

A Monte Carlo study over the three designs, with eleven design--sample-size cells, confirms these findings (\Cref{app:numerics-montecarlo}).  At the balanced concentration, the ratio of the estimated Hellinger risk of the reweighted NPMLE to that of the ordinary NPMLE lies between $\numBalancedRiskMin$ and $\numBalancedRiskMax$ in every cell, the resolved support sizes of the two estimators agree in $\numBalancedSupportAgreement\%$ of the \numPairedDatasets{} paired replications, and the cell means of the ordinary log-likelihood gap, multiplied by $\log n$, stay between $\numBalancedScaledLossMin$ and $\numBalancedScaledLossMax$ across sample sizes and dimensions, as \Cref{thm:main}(iii) predicts.  The Tiny fits agree with the ordinary NPMLE to within the optimization accuracy, whereas the BB-scale fits have risk ratios between $\numBBRiskMin$ and $\numBBRiskMax$.  At the balanced concentration, the density accuracy and resolved support size of the reweighted NPMLE are thus comparable to those of the ordinary NPMLE, while weights that do not concentrate degrade the fit.

\FloatBarrier
\section{Discussion}\label{sec:discussion}

A single, vanishingly small random reweighting of the likelihood thus yields a point estimator of the mixing distribution that, with high probability, has four properties simultaneously: it is unique, it is exactly sparse with polylogarithmically many atoms in every fixed dimension, its ordinary log-likelihood gap is $O(1/\log n)$, and its fitted density attains a Hellinger rate that is parametric up to logarithmic factors (\Cref{thm:main}).  \citet{polyanskiywu2020} established an $O(\log n)$ support bound for the unrestricted ordinary NPMLE in one dimension under sub-Gaussian mixing, a property they called self-regularization.  Almost-sure uniqueness of the unrestricted NPMLE is now known in every dimension \citep{wang2026uniqueness}.  Our contribution is a polylogarithmic support bound for an exact maximizer of a randomly reweighted likelihood, with vanishingly small weight perturbation, over a fixed compact parameter set.

On the other hand, the estimator depends on one auxiliary draw of the weights, which can be fixed and recorded.  Different draws may select different sparse mixing distributions, although, with high probability, each of them has a fitted density in the same shrinking Hellinger neighborhood of $f_{G^*}$.  The support bound also applies only to exact maximizers.  For computed solutions, \Cref{cor:approx} retains the likelihood and Hellinger guarantees, and the number of atoms is assessed through the resolved support size (\Cref{app:numerics-measures}).

The argument suggests several extensions.  The fixed set $K$ could be replaced by the convex hull of the data, which contains the support of every weighted or unweighted NPMLE over $\RR^d$ (\Cref{sec:mixtures}); this requires tracking how the constants depend on its radius, which is of order $\sqrt{\log n}$.  \Cref{thm:effective} applies to any positive continuous kernel, and extending the full result to another mixture model requires a low-rank approximation of the relative fitted vectors, a uniqueness mechanism such as generic evaluation independence, and a Hellinger bound for mixing distributions with nearly maximal likelihood.  Natural candidates include exponential-family mixtures, heteroscedastic Gaussian mixtures \citep{soloffetal2025} and mixtures of regressions \citep{jiangguntuboyina2025}.  The Gamma law of the weights enters the proofs in two places, an exact density-ratio identity in the change of variables of \Cref{lem:determinant-moment} and the concentration of the weights, so other independent weight laws with smooth densities should be amenable to the same argument, with additional remainder terms to control.

\section*{Acknowledgments}

The author would like to thank Adityanand Guntuboyina for many helpful discussions on the NPMLE in general and for his kind encouragement in pursuing the problem of NPMLE sparsity.  The author would also like to thank Xiao Ma for stimulating conversations on AI for mathematics, in particular, on using AI tools to construct and verify mathematical proofs.

\bibliographystyle{plainnat}
\bibliography{paper_v3}

\clearpage
\appendix
\crefalias{section}{appendix}
\crefalias{subsection}{appendix}

\section{Proofs of the main results}
\label{app:main-proofs}

This appendix assembles the proof of \Cref{thm:main} from the results of \Cref{sec:effective,sec:gaussian}, proves \Cref{cor:tiny,cor:weak,cor:approx}, and proves \Cref{prop:dirichlet} together with a Gamma tail bound.  The remaining supporting results are proved in \Cref{app:effective,app:gaussian,app:hellinger}.

\subsection{Proof of Theorem~\ref{thm:main}}\label{app:proof-main}

\begin{proof}[Proof of \Cref{thm:main}]
Fix $b>0$.  All exceptional probabilities below are uniform over $G^*\in\Pcal(K)$.

\emph{Step 1: the data event.}  Under \eqref{eq:model}, the joint distribution of $(X_1,\ldots,X_n)$ has a density, so the generic event $\{X\in\mathcal X_{n,d}\}$ of \Cref{prop:generic} has probability one.  Since $\norm{\Theta_i}\le S$, a Gaussian union bound gives $\max_i\norm{X_i}\le C_b\sqrt{\log n}$ outside an event of probability $O(n^{-b-2})$.  Let $\mathcal E_X$ denote the intersection of these two events.

\emph{Step 2: parts \textup{(ii)} and \textup{(iii)}.}  On $\mathcal E_X$, the dataset satisfies the hypothesis of \Cref{cor:conditional} with $C_0=C_b$.  Hence, with conditional probability at least $1-C_bn^{-b}$ given $X$, every extreme optimizer has at most $C_b\bar r_n$ atoms, the weights satisfy $\max_i\abs{W_i-1}\le C_b\bar r_n^{-1/2}$, and every reweighted NPMLE obeys both bounds in part (iii).  Integrating this conditional bound over $\mathcal E_X$ proves parts (ii) and (iii) outside an event of probability $C_bn^{-b}+O(n^{-b-2})$.

\emph{Step 3: uniqueness and part \textup{(i)}.}  Since $\bar r_n=o(n)$, we have $2C_b\bar r_n(d+1)\le n$ for all sufficiently large $n$.  On the event of Step 2 and on the generic event, the second part of \Cref{prop:generic}, applied with $m=C_b\bar r_n$, shows that the optimizer fiber $\Mcal_X(W)$ is a singleton whose unique member has at most $C_b\bar r_n$ atoms.  This proves part (i).

\emph{Step 4: part \textup{(iv)}.}  On the event of Step 2, part (iii) gives, for all sufficiently large $n$,
\[
 L_n(\Ghat_W)\ge L_n(\Ghat_0)-1\ge L_n(G^*)-1,
\]
because $G^*$ is feasible for the ordinary likelihood.  On the event of \Cref{thm:near-mle}, which has probability at least $1-C_bn^{-b}$, the implication \eqref{eq:near-truth} then applies to $G=\Ghat_W$ and proves part (iv).

\emph{Step 5: probability.}  The exceptional events of Steps 1, 2, and 4 have total probability at most $C_bn^{-b}$ after enlarging $C_b$, which completes the proof.
\end{proof}

\subsection{Proofs of the corollaries}\label{app:proof-corollaries}

\begin{proof}[Proof of \Cref{cor:tiny}]
Take $\alpha_n=n^{2L+2}$ and work on the data event $\mathcal E_X$ of Step 1 in the proof of \Cref{thm:main}.

\emph{Sparsity and uniqueness.}  In \Cref{prop:gaussian-width}, choose $m=\lceil C_L\log n/\log\log n\rceil$ with $C_L$ sufficiently large.  Then $\rank(P_m)=O_L(r_n)$ and $\eps_m\le n^{-L-3}$.  With $q=(b+2)\log n$, both conditions in \eqref{eq:effective-conditions} hold for all large $n$.  Hence \Cref{thm:effective} shows that every extreme optimizer has $O_{b,L}(\bar r_n)$ atoms, and \Cref{prop:generic} again gives uniqueness.

\emph{Size of the perturbation.}  \Cref{lem:gamma-concentration} and a union bound give
\[
 \PP\{\max_i\abs{W_i-1}>t\}\le2n\exp(-\alpha_nt^2/4),
 \qquad 0<t\le1,
\]
and the choice $t=n^{-L}$ yields the bound $2ne^{-n^2/4}$.

\emph{Hellinger accuracy.}  \Cref{thm:effective} also gives an ordinary log-likelihood gap $\Delta L_n(\Ghat_W)\le C_{b,L}\bar r_n/n^{2L+2}=o(1)$.  Hence $L_n(\Ghat_W)\ge L_n(\Ghat_0)-1\ge L_n(G^*)-1$ for all sufficiently large $n$.  Apply \Cref{thm:near-mle} and combine the exceptional events.
\end{proof}

\begin{proof}[Proof of \Cref{cor:weak}]
By \Cref{thm:main}(iv), $H(f_{\Ghat_W},f_{G^*})\to0$ in probability, so every subsequence has a further subsequence along which the convergence holds almost surely.  By compactness of $\Pcal(K)$, we may pass to a further subsequence with $\Ghat_W\Rightarrow G$ for some $G$.  Along it, the Gaussian kernel is bounded and continuous, so $f_{\Ghat_W}(x)\to f_G(x)$ for every $x$, while Hellinger convergence gives $L^1$ convergence to $f_{G^*}$.  Hence $f_G=f_{G^*}$ almost everywhere.  The Fourier transform of $f_G$ is $t\mapsto e^{-\norm{t}^2/2}\chi_G(t)$, where $\chi_G(t)=\int e^{\mathrm i t^\top\theta}\,dG(\theta)$ is the characteristic function of $G$.  Since the Gaussian factor never vanishes, $\chi_G=\chi_{G^*}$ and hence $G=G^*$.  Every subsequential limit is therefore $G^*$.
\end{proof}

\begin{proof}[Proof of \Cref{cor:approx}]
Work on the data event $\mathcal E_X$ of Step 1 in the proof of \Cref{thm:main}, and apply \Cref{thm:effective} with the Gaussian width construction used in \Cref{cor:conditional}, so that $r+q=O(\bar r_n)$ and $\alpha_n=\bar r_n\log n$.  Since $\tau_n\le c_0/\log n$ is below the constant $c$ of \Cref{thm:effective} for large $n$, the bounds \eqref{eq:effective-approx-fit} hold simultaneously for every $\tau_n$-optimal weighted fit on the event of \Cref{thm:effective} and give the first two conclusions.  In particular, every such $\Gtil$ satisfies $L_n(\Gtil)\ge L_n(\Ghat_0)-1\ge L_n(G^*)-1$ for large $n$, and the uniform bound of \Cref{thm:near-mle} gives the third.
\end{proof}

\subsection{Dirichlet representation and Gamma concentration}\label{app:dirichlet}

This subsection proves \Cref{prop:dirichlet} and the Gamma tail bound used in the proofs of \Cref{cor:tiny} and \Cref{cor:conditional}.

\begin{proof}[Proof of \Cref{prop:dirichlet}]
The standard Gamma--Dirichlet change of variables gives
\[
 (\pi_1,\ldots,\pi_n)\sim\mathrm{Dirichlet}(\alpha_n,\ldots,\alpha_n),
 \qquad W_+\sim\mathrm{Gamma}(n\alpha_n,\alpha_n),
\]
with $\pi=(\pi_1,\ldots,\pi_n)$ and $W_+$ independent.  The Dirichlet moments are
\[
 \EE\pi_i=\frac1n,
 \quad
 \operatorname{Var}(\pi_i)=\frac{n-1}{n^2(n\alpha_n+1)},
 \quad
 \operatorname{Cov}(\pi_i,\pi_j)=-\frac1{n^2(n\alpha_n+1)}\quad(i\ne j).
\]
Substitution gives the conditional mean and variance of $F_n^{(\alpha_n)}h=\sum_i\pi_ih(X_i)$.

We have
\[
 \norm{F_n^{(\alpha_n)}-F_n}_{\TV}
 \le\frac12\sum_i\abs{\pi_i-1/n}.
\]
Cauchy--Schwarz and the marginal Dirichlet variance therefore imply
\[
 \EE\bigl(\norm{F_n^{(\alpha_n)}-F_n}_{\TV}\mid X\bigr)
 \le\frac n2\sqrt{\operatorname{Var}(\pi_1)}
 =\frac12\sqrt{\frac{n-1}{n\alpha_n+1}}.
\]

Because $\EE W_i=1$ and $\EE(\pi_i\mid X)=1/n$, independence and integrability give
\[
 \EE_{X,W}\frac1n\sum_iW_i\log f_G(X_i)
 =\EE_{G^*}\log f_G(X_1)
 =\EE_{X,W}\sum_i\pi_i\log f_G(X_i).
\]
For compact $K\subset B(0,S)$,
\[
 e^{-S\norm{x}-S^2/2}
 \le\frac{f_G(x)}{\phi_d(x)}\le e^{S\norm{x}},
\]
so all displayed log densities are integrable under a compactly supported Gaussian mixture.  Their population difference at $G^*$ and $G$ is
\[
 \int f_{G^*}(x)\log\frac{f_{G^*}(x)}{f_G(x)}\,dx
 =\KL(f_{G^*}\|f_G).
\]
\end{proof}

\begin{lemma}[Gamma coordinate concentration]
\label{lem:gamma-concentration}
If $W\sim\mathrm{Gamma}(a,a)$, then for $0<t\le1$,
\[
 \PP\{\abs{W-1}>t\}\le2e^{-at^2/4}.
\]
\end{lemma}

\begin{proof}
For $0<\lambda<a$, $\EE e^{\lambda(W-1)}=e^{-\lambda}(1-\lambda/a)^{-a}$.  Optimizing the Chernoff bound gives $\PP(W\ge1+t)\le\exp\{-a(t-\log(1+t))\}$.  Since $t-\log(1+t)\ge t^2/4$ for $0\le t\le1$, this is at most $e^{-at^2/4}$.  Similarly, $\PP(W\le1-t)\le\exp\{-a(-t-\log(1-t))\}\le e^{-at^2/2}$.  Adding the two tails proves the claim.
\end{proof}

\section{Proof of the effective-dimension theorem}
\label{app:effective}

This appendix proves \Cref{thm:effective} in four steps: basic regularity of the fitted-vector problem (\Cref{app:effective-basic}), the Hessian lower bound and localization of the weighted fit (\Cref{app:effective-local}), the Gamma change of variables (\Cref{app:effective-cov}), and the final probability estimates (\Cref{app:effective-proof}).  Throughout, we use the setting of \eqref{eq:value}--\eqref{eq:width}.

\subsection{Basic convex geometry}
\label{app:effective-basic}

\begin{lemma}[Regularity of the fitted-vector problem]
\label{lem:regularity}
In the setup of \eqref{eq:value}, $\Ccal$ is compact, $\Psi$ is convex and continuously differentiable on $(0,\infty)^n$, and
\[
 \nabla\Psi(w)=z(w)=\log\vhat(w).
\]
The maps $w\mapsto\vhat(w)$ and $w\mapsto z(w)$ are locally Lipschitz.  All maximizing mixing distributions have fitted vector $\vhat(w)$.
\end{lemma}

\begin{proof}
Since $\Acal$ is compact in $\RR^n$, Carath\'eodory's theorem expresses every point of $\conv(\Acal)$ as a convex combination of at most $n+1$ points of $\Acal$; hence $\Ccal$ is compact.  Positivity and compactness give constants $0<\underline v\le\overline v<\infty$ such that $\underline v\le v_i\le\overline v$ for all $v\in\Ccal$ and every $i$.  The objective $v\mapsto\sum_iw_i\log v_i$ is continuous and strictly concave, so it has a unique maximizing fitted vector.

Fix a compact box of weights in the positive orthant and let $a_0>0$ be a common lower bound on its coordinates.  Put $v=\vhat(w)$, $v'=\vhat(w')$, and $f_w(u)=\sum_iw_i\log u_i$.  Adding the variational inequalities for the two maximizations gives
\[
 0\le\{\nabla f_w(v)-\nabla f_{w'}(v')\}^\top(v-v').
\]
Strong concavity of $f_w$ on $[\underline v,\overline v]^n$ and $\nabla f_w(v')-\nabla f_{w'}(v')=(w-w')/v'$ imply
\[
 0\le-\frac{a_0}{\overline v^{\,2}}\norm{v-v'}_2^2
       +\frac1{\underline v}\norm{w-w'}_2\norm{v-v'}_2.
\]
Thus $\vhat$ is locally Lipschitz, and so is its coordinatewise logarithm.

As a maximum of linear functions of $w$, $\Psi$ is convex.  For small $h$, optimality at $w$ and at $w+h$ gives
\[
 h^\top z(w)\le\Psi(w+h)-\Psi(w)\le h^\top z(w+h).
\]
Continuity of $z$ yields differentiability of $\Psi$ with gradient $z$.  Finally, the weighted likelihood depends on a mixing distribution only through its fitted vector in $\Ccal$, which proves the last assertion.
\end{proof}

\begin{lemma}[Extreme optimizer characterization]
\label{lem:extreme}
For every $w>0$, the fiber $\Mcal_A(w)$ in \eqref{eq:fiber} is nonempty, compact, and convex.  A mixing distribution $G\in\Mcal_A(w)$ is extreme if and only if
\[
 G=\sum_{j=1}^kp_j\delta_{\theta_j},\qquad p_j>0,
\]
for finitely many distinct points whose evaluation vectors $A(\theta_1),\ldots,A(\theta_k)$ are linearly independent.  Consequently, every extreme optimizer has at most $n$ atoms.  The function $\sext$ in \eqref{eq:sext} is Borel measurable.
\end{lemma}

\begin{proof}
\emph{Basic properties.}  Probability measures on compact $\Theta$ form a compact convex set under weak convergence, and the continuous moment constraints defining $\Mcal_A(w)$ are closed.  Nonemptiness follows from the representation of $\vhat(w)\in\Ccal$ as a finite convex combination.

\emph{The contact identity.}  Let $\widehat v=\vhat(w)$ and define
\[
 \lambda_i=\frac{w_i}{\widehat v_i\sum_{\ell=1}^nw_\ell},
 \qquad q(\theta)=\lambda^\top A(\theta).
\]
The directional derivative of the weighted objective from $\widehat v$ toward $A(\theta)$ is nonpositive, which gives $q(\theta)\le1$.  If $G$ is in the fiber, then $\int q\,dG=\lambda^\top\widehat v=1$; hence $q=1$ on $\supp(G)$.  Therefore, on an optimizing support, linear dependence and affine dependence of the evaluation vectors are equivalent: applying $\lambda^\top$ to a linear relation shows that its coefficients sum to zero.

\emph{Large support is not extreme.}  Suppose that $G$ has at least $n+1$ support points, including the case of infinite support.  Choose $n+1$ disjoint Borel neighborhoods $D_1,\ldots,D_{n+1}$ of positive $G$-mass.  The vectors $a_j=\int_{D_j}A\,dG\in\RR^n$ are linearly dependent.  For coefficients $c_j$, not all zero, with $\sum_jc_ja_j=0$, define the signed measure $d\nu=\sum_jc_j\bm 1_{D_j}\,dG$.  Then $\int A\,d\nu=0$ and, because $q=1$ $G$-almost surely, $\nu(\Theta)=\int q\,d\nu=0$.  For sufficiently small $t>0$, both $G+t\nu$ and $G-t\nu$ are distinct probability measures in the same fiber, so $G$ is not extreme.

\emph{Finite support.}  If $G$ has finite support but the evaluation vectors are dependent, take a nonzero relation $\sum_jc_jA(\theta_j)=0$.  The contact identity implies $\sum_jc_j=0$, and small perturbations $p_j\pm tc_j$ again decompose $G$ inside the fiber.  Conversely, if the evaluation vectors are linearly independent and $G=tG_1+(1-t)G_2$ with $G_1,G_2$ in the fiber, then nonnegativity forces both measures to be supported on $\supp(G)$.  Their masses solve the same full-rank linear system, so $G_1=G_2=G$.  This proves the characterization and the $n$-atom bound.

\emph{Measurability.}  Let $\mathcal W_k$ be the set of weights admitting an extreme optimizer with exactly $k$ atoms.  For integers $j\ge2$, consider the compact set of $(w,\theta_1,\ldots,\theta_k,p)$ in $[1/j,j]^n\times\Theta^k\times\Delta_k$ satisfying
\[
 p_\ell\ge1/j,\qquad
 \sum_{\ell=1}^kp_\ell A(\theta_\ell)=\vhat(w),
 \qquad \det(V^\top V)\ge1/j,
\]
where $V$ has columns $A(\theta_\ell)$.  Its projection onto $w$ is compact.  Taking the countable union over $j$ gives $\mathcal W_k$.  Hence $\{w:\sext(w)\ge k\}=\bigcup_{\ell=k}^n\mathcal W_\ell$ is Borel for each $k$.
\end{proof}

\subsection{Support directions and radial localization}
\label{app:effective-local}

We first prove the support--sensitivity inequality stated in \Cref{sec:effective}.

\begin{proof}[Proof of \Cref{lem:hessian-summary}]
By \Cref{lem:regularity}, $z=\nabla\Psi$ is locally Lipschitz, so Rademacher's theorem gives differentiability almost everywhere.  At those points its derivative is the symmetric positive semidefinite Hessian of the convex function $\Psi$.

Fix such a point and an extreme optimizer $G=\sum_{j=1}^kp_j\delta_{\theta_j}$.  Write $A_j=A(\theta_j)$ and $\widehat v=\vhat(w)$.  Varying the masses while preserving their sum generates the relative fitted-vector subspace
\[
 \mathcal S=\left\{
 \left(\frac{\sum_jc_jA_{ij}}{\widehat v_i}\right)_{i=1}^n:
 \sum_jc_j=0\right\}.
\]
By \Cref{lem:extreme}, the evaluation vectors are linearly, hence affinely, independent; therefore $\dim(\mathcal S)=k-1$.  Every sufficiently small element of $\mathcal S$ is feasible in both signs, by changing only the positive masses.  First-order optimality gives $w^\top u=0$ for $u\in\mathcal S$.

Let $P_G$ be the orthogonal projection onto $\diag(w)^{1/2}\mathcal S$.  For small $h$, choose the feasible relative variation $u=\diag(w)^{-1/2}P_G\diag(w)^{-1/2}h$, that is, the candidate fitted vector $\widehat v(\one+u)$.  Expanding its objective at weight $w+h$ gives
\[
 \Psi(w+h)\ge\Psi(w)+h^\top z(w)
 +h^\top u-\frac12u^\top\diag(w)u+o(\norm{h}_2^2).
\]
Because $w^\top u=0$, the two quadratic terms combine to
\[
 \frac12h^\top\diag(w)^{-1/2}P_G\diag(w)^{-1/2}h.
\]
Comparison with the second-order expansion of $\Psi$ at $w$ yields $\nabla^2\Psi(w)\succeq\diag(w)^{-1/2}P_G\diag(w)^{-1/2}$, and conjugation by $\diag(w)^{1/2}$ proves the claim.  The exceptional set is the nondifferentiability set of $z$, which does not depend on the optimizer.  Choosing an extreme optimizer that attains $\sext(w)$ gives the final statement.
\end{proof}

For the localization argument, retain \eqref{eq:relative-set}--\eqref{eq:width} and write
\[
 \xi=w-\one,\qquad \omega_P=\norm{P\xi}_2,\qquad \omega=\norm{\xi}_2,
 \qquad \rho=\frac1{16}.
\]
Here $\omega_P$ is the part of the weight perturbation that is visible in the effective subspace and $\omega$ is its total size.  For $u\in\Ucal$, let $t(u)=\log(\one+u)$ and $F_w(u)=w^\top t(u)$, so that $F_w(u)=w^\top\log v-w^\top\log v_0$ when $u=v/v_0-\one$.  For the exact fit, write $\widehat u(w)=\vhat(w)/v_0-\one$, $t(w)=t(\widehat u(w))$, and $Z(w)=\xi^\top t(w)$; these agree with the definitions in \Cref{sec:effective}.

\begin{lemma}[Radial localization]
\label{lem:radial}
Let
\begin{equation}\label{eq:localization-event}
 E=\left\{w>0:
\norm{w-\one}_\infty\le\frac1{16},\quad
 \omega_P\le\frac{\rho}{24},\quad \eta\,\omega\le\frac{\rho^2}{24}\right\}.
\end{equation}
For $w\in E$, every $0\le\tau<\rho^2/12$ and every $\tau$-optimal fit with relative vector $u=\vtil/v_0-\one$ satisfy $\norm{u}_2<\rho$ and
\begin{align}
 \norm{u}_2^2&\le36\omega_P^2+12\eta\omega+12\tau,\label{eq:radial-u}\\
 0\le-\one^\top\log(\one+u)
 &\le37\omega_P^2+13\eta\omega+13\tau,\label{eq:radial-loss}\\
 \norm{\log(\one+u)}_2^2
 &\le72\omega_P^2+24\eta\omega+24\tau.\label{eq:radial-log}
\end{align}
For the exact fit, $0\le Z(w)\le37\omega_P^2+13\eta\omega$.  The set $E$ is compact and $\norm{t(w)}_\infty<1/8$ on $E$.
\end{lemma}

\begin{proof}
\emph{A quadratic upper bound.}  Equal-weight optimality at $v_0$ implies $\one^\top u\le0$ for every $u\in\Ucal$.  If $\norm{u}_2\le1/4$, then $\log(1+s)\le s-s^2/3$ coordinatewise.  On $E$, every $w_i\ge15/16$, so
\[
 F_w(u)\le \one^\top u+\xi^\top u-\frac{15}{48}\norm{u}_2^2
 \le \xi^\top u-\frac16\norm{u}_2^2
 \le \omega_P\norm{u}_2+\eta\,\omega-\frac16\norm{u}_2^2,
\]
where the last step splits $\xi^\top u=(P\xi)^\top u+\xi^\top(I-P)u$ and uses \eqref{eq:width}.

\emph{Localization.}  On the sphere $\norm{u}_2=\rho$, the conditions defining $E$ give $F_w(u)\le\rho^2/24+\rho^2/24-\rho^2/6=-\rho^2/12$ for every feasible $u$.  A $\tau$-optimal fit has $F_w(u)\ge-\tau$ because $u=0$ is feasible.  If $\norm{u}_2\ge\rho$, convexity of $\Ucal$ and concavity of $F_w$ allow radial contraction of $u$ to the sphere without decreasing $F_w$ below $-\tau$, contradicting $\tau<\rho^2/12$.  Hence $\norm{u}_2<\rho$.

\emph{Quantitative bounds.}  Let $y=\norm{u}_2$.  The quadratic upper bound and $F_w(u)\ge-\tau$ give $y^2\le6(\omega_Py+\eta\omega+\tau)$.  The bound $6\omega_Py\le y^2/2+18\omega_P^2$ proves \eqref{eq:radial-u}.  For $\abs{u_i}<\rho$,
\[
 \abs{\log(1+u_i)}\le\frac{\abs{u_i}}{1-\rho},
 \qquad
 \abs{\log(1+u_i)-u_i}\le\frac{u_i^2}{2(1-\rho)}\le u_i^2.
\]
The first inequality and \eqref{eq:radial-u} prove \eqref{eq:radial-log}.  Moreover,
\[
 \xi^\top\log(\one+u)
 \le \omega_Py+\eta\,\omega+\frac1{16}y^2
 \le37\omega_P^2+13\eta\omega+12\tau,
\]
where the constants are enlarged after applying Young's inequality and \eqref{eq:radial-u}.  Since $F_w(u)=\one^\top\log(\one+u)+\xi^\top\log(\one+u)\ge-\tau$, and equal-weight optimality gives $-\one^\top\log(\one+u)\ge0$, we obtain \eqref{eq:radial-loss}.  For the exact fit, $F_w(\widehat u)\ge0$ gives $Z(w)\ge-\one^\top t(w)\ge0$ together with the stated upper bound.

\emph{Compactness.}  The set $E$ is closed inside $[15/16,17/16]^n$, hence compact.  For the exact fit, $\norm{t(w)}_\infty\le\norm{\widehat u(w)}_2/(1-\rho)<1/15<1/8$.
\end{proof}

\subsection{The Gamma change of variables}
\label{app:effective-cov}

\begin{lemma}[Determinant moment]
\label{lem:determinant-moment}
Let $E$ be the event in \eqref{eq:localization-event}.  For independent $W_i\sim\mathrm{Gamma}(\alpha,\alpha)$,
\begin{equation}\label{eq:determinant-moment}
 \EE\left[\bm 1_{\{W\in E\}}
 \left(\frac{17}{9}\right)^{\sext(W)-1}
 \exp\{-\alpha[Z(W)+\norm{t(W)}_2^2]\}\right]\le1.
\end{equation}
\end{lemma}

\begin{proof}
\emph{Injectivity.}  On the open set $U=\{w>0:1+t_i(w)>0\ \text{for all }i\}$, define the map $\mathcal F$ in \eqref{eq:change-map}.  It is locally Lipschitz by \Cref{lem:regularity}.  Suppose $\mathcal F(w)=\mathcal F(w')=y$.  Then $w_i=y_i/\{1+t_i(w)\}$, $w_i'=y_i/\{1+t_i(w')\}$, and, since $z(w)-z(w')=t(w)-t(w')$,
\[
 (w-w')^\top\{z(w)-z(w')\}
 =-\sum_i\frac{y_i\{t_i(w)-t_i(w')\}^2}
 {[1+t_i(w)][1+t_i(w')]}\le0.
\]
The gradient of the convex function $\Psi$ is monotone, so the left side is nonnegative.  Hence $t(w)=t(w')$, and then $w=w'$.  Thus $\mathcal F$ is injective on $U$.

\emph{Jacobian lower bound.}  At a point where the Hessian exists, let
\[
 D_t=\diag(\one+t(w)),\qquad
 R_w=\diag(w)^{1/2}\nabla^2\Psi(w)\diag(w)^{1/2}.
\]
Since $D_t$ is diagonal, the derivative $D\mathcal F=D_t+\diag(w)\nabla^2\Psi(w) =\diag(w)^{1/2}(D_t+R_w)\diag(w)^{-1/2}$ is similar to $D_t+R_w$.  By \Cref{lem:hessian-summary}, $R_w\succeq P_*$ for an orthogonal projection $P_*=VV^\top$ of rank $\sext(w)-1$, where $V$ has orthonormal columns.  On $E$, \Cref{lem:radial} gives $7I/8\preceq D_t\preceq9I/8$, so $V^\top D_t^{-1}V\succeq8I/9$.  Monotonicity of the determinant on positive definite matrices, together with the fact that $D_t^{-1/2}VV^\top D_t^{-1/2}$ and $V^\top D_t^{-1}V$ have the same nonzero eigenvalues, yields
\[
 \det(D_t+R_w)\ge\det(D_t+VV^\top)
 =\det(D_t)\det(I+V^\top D_t^{-1}V),
\]
and therefore
\begin{equation}\label{eq:jacobian-lower}
 \det D\mathcal F(w)
 \ge\prod_i\{1+t_i(w)\}
 \left(\frac{17}{9}\right)^{\sext(w)-1}>0.
\end{equation}

\emph{Density ratio.}  Let $p_\alpha$ be the product Gamma density.  Since $\mathcal F_i(w)=w_i\{1+t_i(w)\}$,
\[
 \frac{p_\alpha(\mathcal F(w))}{p_\alpha(w)}
 =\prod_i\{1+t_i(w)\}^{\alpha-1}
   \exp\left\{-\alpha\sum_iw_it_i(w)\right\}.
\]
Multiply by \eqref{eq:jacobian-lower}, write $\sum_iw_it_i=\one^\top t+Z$, and use $\log(1+t)\ge t-t^2$ for $\abs{t}\le1/8$.  This gives, almost everywhere on $E$,
\[
 \frac{p_\alpha(\mathcal F(w))\det D\mathcal F(w)}{p_\alpha(w)}
 \ge\left(\frac{17}{9}\right)^{\sext(w)-1}
 e^{-\alpha\{Z(w)+\norm{t(w)}_2^2\}}.
\]

\emph{Area formula.}  The set $E$ is compact and lies inside $U$.  The injective area formula for Lipschitz maps \citep[Chapter 3]{evansgariepy2015} therefore gives
\[
 \int_Ep_\alpha(\mathcal F(w))\abs{\det D\mathcal F(w)}\,dw
 =\int_{\mathcal F(E)}p_\alpha(y)\,dy\le1.
\]
The determinant is positive almost everywhere on $E$.  Integrating the previous lower bound against $p_\alpha$ over $E$ proves \eqref{eq:determinant-moment}.
\end{proof}

\subsection{Proof of the effective-dimension theorem}
\label{app:effective-proof}

\begin{proof}[Proof of \Cref{thm:effective}]
\emph{Concentration of the visible perturbation.}  Let $\xi=W-\one$ and put $N_q=r\log5+q$.  The Gamma moment-generating function implies that, for every deterministic unit vector $v$ and $\abs{\lambda}\le\sqrt\alpha/2$,
\[
 \log\EE e^{\lambda\sqrt\alpha\,v^\top\xi}
 =\alpha\sum_i\left[-\log\left(1-\frac{\lambda v_i}{\sqrt\alpha}\right)
 -\frac{\lambda v_i}{\sqrt\alpha}\right]\le\lambda^2.
\]
Consequently, $\PP\{\sqrt\alpha\,v^\top\xi>u\}\le e^{-u^2/4}$ for $0\le u\le\sqrt\alpha$.  A $1/2$-net of the unit sphere in the $r$-dimensional range of $P$ has at most $5^r$ points and satisfies $\norm{y}_2\le2\max_v v^\top y$ for $y$ in that range.  Hence
\begin{equation}\label{eq:proj-concentration}
 \PP\{\alpha\norm{P\xi}_2^2>16N_q\}\le e^{-q}
\end{equation}
provided $\alpha\ge4N_q$.

\emph{Coordinatewise concentration.}  The scalar Gamma Chernoff bound gives
\[
 \PP\{\norm{\xi}_\infty>2\sqrt{Q/\alpha}\}
 \le2n e^{-Q}=e^{-q}
\]
when $\alpha\ge4Q$.  Let $E_q$ be the intersection of these two events.  Then $\PP(E_q^c)\le2e^{-q}$ and, on $E_q$,
\[
 \omega_P\le4\sqrt{N_q/\alpha},\qquad
 \omega\le2\sqrt{nQ/\alpha},\qquad
 \eta\,\omega\le2/\alpha,
\]
where the last inequality uses \eqref{eq:effective-conditions}.  Taking the universal constant in \eqref{eq:effective-conditions} sufficiently large ensures $E_q\subset E$, with $E$ as in \eqref{eq:localization-event}.

\emph{Likelihood bounds.}  Apply \Cref{lem:radial}.  Since $N_q\le C(r+q)$, its exact-fit bounds yield
\[
 0\le\one^\top\log v_0-\one^\top\log\vhat(W)
 \le C\frac{r+q}{\alpha},
 \qquad
 \norm{\log\{\vhat(W)/v_0\}}_2^2
 \le C\frac{r+q}{\alpha}.
\]
The same lemma gives \eqref{eq:effective-approx-fit} simultaneously for all $\tau<c$, after taking $c=\rho^2/12$ and enlarging $C$.

\emph{Support bound.}  Adding the exact-fit bounds of \Cref{lem:radial} for $Z(W)$ and for $\norm{t(W)}_2^2$ gives, on $E_q$,
\[
 \alpha\{Z(W)+\norm{t(W)}_2^2\}
 \le109\alpha\omega_P^2+37\alpha\eta\omega
 \le1744N_q+74=:\Lambda_q.
\]
Set $\kappa=\log(17/9)$.  By \Cref{lem:determinant-moment} and Markov's inequality,
\[
 \PP\{E_q,\ (\sext(W)-1)\kappa>\Lambda_q+q\}\le e^{-q}.
\]
Adding $\PP(E_q^c)$ shows that, outside probability $3e^{-q}$,
\[
 \sext(W)
 \le1+\frac{1744N_q+74+q}{\kappa}
 \le C(r+q).
\]
All conclusions hold on the same event.
\end{proof}

\section{Gaussian approximation and uniqueness}
\label{app:gaussian}

\subsection{Proof of the Gaussian residual-width bound}

\begin{proof}[Proof of \Cref{prop:gaussian-width}]
Set $\gamma_i=\phi_d(x_i)/v_{0,i}$.  The Gaussian identity gives
\begin{equation}\label{eq:gaussian-ratio}
 \frac{\phi_d(x_i-\theta)}{v_{0,i}}
 =\gamma_i e^{-\norm{\theta}^2/2}e^{x_i^\top\theta}.
\end{equation}
Because
\[
 \frac{v_{0,i}}{\phi_d(x_i)}
 =\int e^{x_i^\top\vartheta-\norm{\vartheta}^2/2}\,dG_0(\vartheta)
 \ge e^{-ST-S^2/2},
\]
we have $\gamma_i\le e^{ST+S^2/2}$.  For $\abs{s}\le ST$, Taylor's theorem yields
\[
 \left|e^s-\sum_{j=0}^m\frac{s^j}{j!}\right|
 \le e^{ST}\frac{(ST)^{m+1}}{(m+1)!}.
\]
Applying this with $s=x_i^\top\theta$ and multiplying by the remaining factors in \eqref{eq:gaussian-ratio} shows that
\[
 \gamma_i e^{-\norm{\theta}^2/2}
 \sum_{j=0}^m\frac{(x_i^\top\theta)^j}{j!}
\]
approximates $\phi_d(x_i-\theta)/v_{0,i}$ with coordinatewise error at most $\eps_m$ in \eqref{eq:eps-m}, uniformly over $i$ and $\theta\in K$.

The multinomial expansion of $(x_i^\top\theta)^j$ uses only monomials $x_i^\beta$ with multi-index $\abs{\beta}\le m$.  Hence every truncated vector lies in the span of
\[
 \left\{(\gamma_ix_i^\beta)_{i=1}^n:\abs{\beta}\le m\right\},
\]
whose dimension is at most $\binom{m+d}{d}$.  Adjoin $\one$ to this span and let $P_m$ be the orthogonal projection onto the resulting subspace; its rank obeys \eqref{eq:rank-m}.  Integrating the coordinatewise approximation under any $G\in\Pcal(K)$ preserves its error, so $f_G(x)/v_0$ lies within coordinatewise distance $\eps_m$, hence Euclidean distance $\sqrt n\,\eps_m$, of the range of $P_m$.  Since $\one$ also lies in that range and orthogonal projection minimizes Euclidean distance to the subspace, \eqref{eq:eta-m} follows.  If $ST=0$, the exponential term is constant in the relevant inner product and the expansion is exact.
\end{proof}

\subsection{Proof of the conditional Gaussian corollary}

\begin{proof}[Proof of \Cref{cor:conditional}]
Choose any ordinary NPMLE $\Ghat_0$ and let $v_0=(f_{\Ghat_0}(x_i))_{i=1}^n$; by \Cref{lem:regularity}, this vector does not depend on the choice of $\Ghat_0$.  Apply \Cref{prop:gaussian-width} with
\[
 m=\left\lceil C_m\frac{\log n}{\log\log n}\right\rceil,
\]
where $C_m>8$ is a sufficiently large constant depending only on the fixed parameters.  The factorial inequality $(m+1)!\ge\{(m+1)/e\}^{m+1}$ and $T\le C_0\sqrt{\log n}$ give
\begin{align*}
 \log\eps_m
 &\le 2ST+S^2/2+(m+1)\log\frac{eST}{m+1}\\
 &\le-3\log n
\end{align*}
for all sufficiently large $n$.  Indeed, $\log(ST)\le\tfrac12\log\log n+O(1)$, whereas $\log(m+1)=\log\log n-\log\log\log n+O(1)$.  If $ST=0$, the error is zero.  Thus
\[
 r:=\rank(P_m)=O(r_n),\qquad
 \eta:=\eta_m\le n^{-5/2}.
\]

Take $q=(b+2)\log n$ and $Q=q+\log(2n)$.  Then $r+Q=O_b(\bar r_n)$ and
\[
 \frac{\alpha_n}{r+Q}\asymp_b\log n\to\infty.
\]
Moreover,
\[
 \eta\sqrt{\alpha_n nQ}
 \le n^{-5/2}\sqrt{\bar r_n\log n\cdot nQ}=o(1).
\]
The hypotheses of \Cref{thm:effective} therefore hold.  Its failure probability is at most $3n^{-b-2}$, and it yields $\sext(W)\le C_b\bar r_n$ together with the likelihood and fitted-vector conclusions.

By \Cref{lem:gamma-concentration},
\begin{equation}\label{eq:gamma-coordinate}
 \PP\{\abs{W_i-1}>t\}\le2e^{-\alpha_nt^2/4},
 \qquad 0<t\le1.
\end{equation}
Choose $t=C_b'/\sqrt{\bar r_n}$ with $(C_b')^2/4\ge b+3$.  Since $\alpha_nt^2=(C_b')^2\log n$, a union bound gives
\[
 \PP\{\max_i\abs{W_i-1}>t\}
 \le2n^{1-(C_b')^2/4}\le2n^{-b-2}.
\]
Intersect the events and enlarge constants.
\end{proof}

\subsection{Proof of generic evaluation independence}
\label{app:generic}

\begin{proof}[Proof of \Cref{prop:generic}]
\emph{Parametrizing dependences.}  Fix an integer $k$ satisfying $k(d+1)\le n$.  A nonzero linear dependence among $k$ Gaussian evaluation columns has a coefficient vector with at least one nonzero coordinate.  Normalize that coordinate to one; there are only $k$ possible choices.  In one such chart, collect the $kd$ location coordinates and the remaining $k-1$ coefficients in a parameter $\zeta\in\mathcal O\subset\RR^p$, where
\[
 p=kd+k-1=k(d+1)-1<n.
\]
The domain $\mathcal O$ is open after imposing distinctness of the locations.  For the normalized coefficient vector $(c_1,\ldots,c_k)$, define
\[
 g_\zeta(y)=\sum_{j=1}^kc_j\phi_d(y-\theta_j),\qquad y\in\RR^d.
\]
This function is jointly real analytic in $(\zeta,y)$.

\emph{Nontriviality.}  For every admissible $\zeta$, $g_\zeta$ is not identically zero.  Choose $v\in\RR^d$ such that the scalars $v^\top\theta_j$ are distinct; the excluded $v$ lie in finitely many proper hyperplanes.  If $g_\zeta(sv)=0$ for every $s\in\RR$, then after division by the positive common Gaussian factor,
\[
 \sum_{j=1}^kc_j e^{-\norm{\theta_j}^2/2}e^{s v^\top\theta_j}=0
 \qquad\text{for all }s.
\]
Differentiating at $s=0$ for orders $0,\ldots,k-1$ gives a Vandermonde system with distinct nodes $v^\top\theta_j$, forcing every coefficient to vanish.  This contradicts the normalization.

\emph{Dimension count.}  Consider the incidence set
\[
 I=\{(\zeta,x_1,\ldots,x_n)\in \mathcal O\times(\RR^d)^n:
        g_\zeta(x_i)=0\text{ for all }i\}.
\]
At an incidence point, analyticity and nontriviality imply that for each $i$ there is a multi-index $\mu_i$ of minimal total order with $\partial_y^{\mu_i}g_\zeta(x_i)\ne0$.  Its order is at least one because $g_\zeta(x_i)=0$.  Choose a coordinate $\ell_i$ with $\mu_{i,\ell_i}>0$ and put $\mu_i'=\mu_i-e_{\ell_i}$.  By minimality,
\[
 h_i(\zeta,x_i):=\partial_y^{\mu_i'}g_\zeta(x_i)=0,
 \qquad
 \partial_{x_{i,\ell_i}}h_i(\zeta,x_i)\ne0.
\]
For each fixed choice of the countably many derivative indices, restrict to the open region on which the displayed derivatives are nonzero.  The Jacobian of $(h_1,\ldots,h_n)$ with respect to the selected data coordinates $(x_{1,\ell_1},\ldots,x_{n,\ell_n})$ is diagonal and nonsingular.  The implicit-function theorem therefore places the corresponding part of $I$ in a smooth submanifold of codimension $n$ in $\mathcal O\times\RR^{nd}$, of dimension
\[
 nd+p-n<nd.
\]
Each such submanifold has a countable atlas.  On compact subsets of a chart, projection to the $nd$ data coordinates is Lipschitz, and the image of a bounded set of dimension $D<nd$ under a Lipschitz map has zero $nd$-dimensional Lebesgue measure.  Thus the projection of $I$ onto the data coordinates is contained in a countable union of null sets.

Taking the finite union over coefficient normalizations and over admissible $k$ produces a Borel null set $N_{n,d}$.  Let $\mathcal X_{n,d}=(\RR^d)^n\setminus N_{n,d}$.  This proves the simultaneous column-independence assertion.

For the final statement, suppose every extreme optimizer has at most $m$ atoms and $2m(d+1)\le n$.  If two distinct extreme optimizers existed, their difference would be a nonzero signed measure on at most $2m$ distinct locations with zero Gaussian evaluation vector.  After deleting zero coefficients, this would contradict the independence just established.  Hence the compact convex optimizer fiber has at most one extreme point.  By the Krein--Milman theorem it is the closed convex hull of its nonempty set of extreme points, so it is a singleton.  Its unique member is extreme and has at most $m$ atoms.
\end{proof}

\section{Likelihood nets and Hellinger transfer}
\label{app:hellinger}

\subsection{A finite Gaussian-mixture likelihood net}

\begin{lemma}[Simultaneous Hellinger and log-likelihood net]
\label{lem:finite-net}
Suppose that \Cref{ass:standing} holds and fix $b>0$.  For all sufficiently large $n$, there are a deterministic finite set $\Ncal_n\subset\{f_G:G\in\Pcal(K)\}$ and a deterministic $T_n=O(\sqrt{\log n})$ such that
\begin{equation}\label{eq:net-size}
 \log\abs{\Ncal_n}
 \le C\frac{(\log n)^{d+1}}{(\log\log n)^d},
\end{equation}
\begin{equation}\label{eq:radius-prob}
 \sup_{G^*\in\Pcal(K)}
 \PP_{G^*}\left\{\max_{i\le n}\norm{X_i}>T_n\right\}
 \le n^{-b-2},
\end{equation}
and every $G\in\Pcal(K)$ admits $p\in\Ncal_n$ satisfying
\begin{equation}\label{eq:net-approx}
 H(p,f_G)\le n^{-3},
 \qquad
 \sup_{\norm{x}\le T_n}\abs{\log p(x)-\log f_G(x)}\le n^{-3}.
\end{equation}
Constants are uniform over $G$ and $G^*$.
\end{lemma}

\begin{proof}
The construction has three stages: a radius, a moment-matching mixing distribution with few atoms, and a discretization of its locations and masses.

\emph{Radius.}  For $n\ge3$, take
\[
 T_n=S+\sqrt{2d\log(2d n^{b+3})}.
\]
Under any $G^*\in\Pcal(K)$, $X_i=\Theta_i+Z_i$ with $\norm{\Theta_i}\le S$.  The standard Gaussian coordinate bound and a union bound over $nd$ coordinates prove \eqref{eq:radius-prob}.  If $S=0$, the class contains only $\phi_d$, so assume $S>0$.

\emph{Moment matching.}  Let
\[
 m=\left\lceil C_\star\frac{\log n}{\log\log n}\right\rceil,
 \qquad k_m=\binom{2m+d}{d},
\]
with a fixed constant $C_\star>16$.  Consider the vector of all nonconstant monomials in $\theta$ of total degree at most $2m$.  Its image of $K$ is compact, and the vector of its integrals under any $G\in\Pcal(K)$ lies in the corresponding convex hull.  Carath\'eodory's theorem gives a mixing distribution $G_m$ with at most $k_m$ atoms that matches every moment of $G$ of total degree at most $2m$.

\emph{Local relative approximation.}  For $\norm{x}\le T_n$, set $\Lambda_n=ST_n+S^2/2$.  The exponent $x^\top\theta-\norm{\theta}^2/2$ has absolute value at most $\Lambda_n$ on $K$, and $f_G(x)/\phi_d(x)\ge e^{-\Lambda_n}$.  Its Taylor polynomial through order $m$ has total degree at most $2m$ in $\theta$, so its integrals under $G$ and $G_m$ coincide.  Bounding both remainders yields
\begin{equation}\label{eq:relative-approx}
 \sup_{\norm{x}\le T_n}
 \left|\frac{f_{G_m}(x)}{f_G(x)}-1\right|
 \le2e^{2\Lambda_n}\frac{\Lambda_n^{m+1}}{(m+1)!}\le n^{-8}
\end{equation}
for all large $n$.  The last inequality follows because $\Lambda_n=O(\sqrt{\log n})$ and $(m+1)!\ge\{(m+1)/e\}^{m+1}$.

\emph{Global approximation.}  Moment matching also controls the global error.  Let $\Delta=G-G_m$.  Completing the square shows
\[
 \int_{\RR^d}
 \frac{\phi_d(x-\theta)\phi_d(x-\vartheta)}{\phi_d(x)}\,dx
 =e^{\theta^\top\vartheta}.
\]
The exponential series is uniformly convergent on $K\times K$.  Every term of degree at most $2m$ integrates to zero against $\Delta\otimes\Delta$, and the total variation measure $\abs{\Delta}$ of $\Delta$ satisfies $\abs{\Delta}(K)\le2$.  Therefore
\begin{align}
 Q_m:=\int\frac{(f_G-f_{G_m})^2}{\phi_d}
 &=\iint e^{\theta^\top\vartheta}\,d\Delta(\theta)\,d\Delta(\vartheta)\notag\\
 &\le4\sum_{j>2m}\frac{S^{2j}}{j!}
 \le4e^{S^2}\frac{S^{4m+2}}{(2m+1)!}\le n^{-32}.
 \label{eq:Qm}
\end{align}
By Cauchy--Schwarz, $H^2(f_G,f_{G_m})\le\norm{f_G-f_{G_m}}_1\le Q_m^{1/2}$, so $H(f_G,f_{G_m})\le n^{-8}$.

\emph{Discretization.}  It remains to discretize the locations and masses.  Put $\eps=n^{-16}$.  Partition a fixed bounding cube into cells of diameter at most $\eps$ and choose one point of $K$ from every cell that meets $K$.  The resulting location net has size $J_n\le C_Kn^{16d}$.  Write $G_m=\sum_{j=1}^{k_m}\lambda_j\delta_{\theta_j}$, padding with zero masses if needed, and round every location to the net.  Let $M=\lceil k_mn^{16}\rceil$.  For $j<k_m$, put $\lambda_j'=\lfloor M\lambda_j\rfloor/M$, and let $\lambda_{k_m}'=1-\sum_{j<k_m}\lambda_j'$.  Then $\sum_j\abs{\lambda_j'-\lambda_j}\le2\eps$.  Let $G'$ be the resulting discrete mixing distribution and $p=f_{G'}$.

Moving a location by at most $\eps$ changes its Gaussian log kernel on $\{\norm{x}\le T_n\}$ by at most $(T_n+S)\eps$.  Weight rounding changes the mixture by relative amount at most $2\eps e^{2\Lambda_n}$.  Combining these estimates with \eqref{eq:relative-approx}, and using $\abs{\log(1+u)}\le2\abs{u}$ for $\abs{u}\le1/2$, gives
\[
 \sup_{\norm{x}\le T_n}\abs{\log p(x)-\log f_G(x)}\le n^{-3}
\]
for all large $n$.  Globally, the fundamental theorem of calculus and $\int\norm{\nabla\phi_d(x)}\,dx<\infty$ give an $L^1$ error $O_d(\eps)$ from location rounding; weight rounding adds at most $2\eps$.  The Hellinger triangle inequality therefore gives $H(p,f_G)\le n^{-3}$.

\emph{Cardinality.}  Let $\Ncal_n$ contain the densities generated by every ordered list of $k_m$ location-net points and every simplex-grid mass vector with denominator $M$.  Then each $G$ has a point satisfying \eqref{eq:net-approx}, and
\[
 \log\abs{\Ncal_n}
 \le k_m\{\log J_n+\log(M+1)\}
 =O(k_m\log n)
 =O\left\{\frac{(\log n)^{d+1}}{(\log\log n)^d}\right\}.
\]
This proves \eqref{eq:net-size}.
\end{proof}

\subsection{Proof of the uniform Hellinger bound}

\begin{proof}[Proof of \Cref{thm:near-mle}]
Fix $G^*\in\Pcal(K)$ and write $p_*=f_{G^*}$.  Let $\Ncal_n,T_n$ be as in \Cref{lem:finite-net}, and let $E_n=\{\max_i\norm{X_i}\le T_n\}$.  On $E_n$, for any $G$ satisfying the hypothesis of \eqref{eq:near-truth}, choose $p\in\Ncal_n$ as in \eqref{eq:net-approx}.  Then
\begin{equation}\label{eq:net-lr-lower}
 \sum_{i=1}^n\log\frac{p(X_i)}{p_*(X_i)}
 \ge-1-n^{-2}\ge-2.
\end{equation}

For deterministic $p\in\Ncal_n$, independence gives
\begin{align*}
 \EE_{p_*}\exp\left\{\frac12\sum_{i=1}^n
       \log\frac{p(X_i)}{p_*(X_i)}\right\}
 &=\left(\int\sqrt{pp_*}\right)^n\\
 &=\left\{1-\frac12H^2(p,p_*)\right\}^n
 \le e^{-nH^2(p,p_*)/2}.
\end{align*}
Markov's inequality therefore yields
\begin{equation}\label{eq:fixed-net-test}
 \PP_{p_*}\left\{\sum_i\log\frac{p(X_i)}{p_*(X_i)}\ge-2\right\}
 \le\exp\left\{1-\frac n2H^2(p,p_*)\right\}.
\end{equation}

Set
\[
 A_n=\log\abs{\Ncal_n}+(b+2)\log n+2,
 \qquad t_n=\sqrt{2A_n/n}.
\]
Summing \eqref{eq:fixed-net-test} over all net points with $H(p,p_*)\ge t_n$ shows that, outside an event of probability at most $e^{-1}n^{-b-2}$, every such point has log-likelihood ratio below $-2$.

Intersect this event with $E_n$.  The net point associated with any such $G$ satisfies \eqref{eq:net-lr-lower}, so it must obey $H(p,p_*)<t_n$.  Hence, simultaneously for all $G$ satisfying the hypothesis of \eqref{eq:near-truth},
\[
 H^2(f_G,p_*)\le(t_n+n^{-3})^2
 \le\frac{4A_n}{n}+2n^{-6}.
\]
The entropy bound \eqref{eq:net-size} gives the claimed rate.  The radius and likelihood-test exceptional probabilities are uniform over $G^*\in\Pcal(K)$, completing the proof.
\end{proof}

\section{Numerical implementation and replication details}
\label{app:numerics}

This appendix describes the designs, the computation, the performance measures and the accuracy checks behind \Cref{sec:numerics}, and reports the full results of the regularization path and the Monte Carlo study.

\begin{table}[tb]
\centering
\caption{Monte Carlo comparison with the ordinary NPMLE at the largest sample size in each design.  ``Balanced'' uses \eqref{eq:scales}, ``Tiny'' uses $\alpha_n=n^3$, and ``BB-scale'' uses $\alpha=1$.  Entries are Monte Carlo means with standard errors in parentheses.  ``Support'' is the resolved support size, and $\Delta L_n$ is the ordinary log-likelihood gap \eqref{eq:gap-discrepancy} of the fit.}
\label{tab:numerical-summary}
\footnotesize
\setlength{\tabcolsep}{3.2pt}
\begin{tabular}{@{}llrrrr@{}}
\toprule
Scenario ($n$) & Method & Support & $H^2(\widehat f,f^*)$ & $H^2(\widehat f_W,\widehat f_0)$ & $\Delta L_n$ \\
\midrule
Three-point, $d=1$ (2000) & Ordinary & $3.98\,(0.07)$ & $6.90\,(0.34)\times 10^{-4}$ & $0$ & $0$ \\
 & Balanced & $4.00\,(0.07)$ & $6.91\,(0.35)\times 10^{-4}$ & $8.31\,(0.50)\times 10^{-6}$ & $0.034\,(0.002)$ \\
 & Tiny & $3.98\,(0.07)$ & $6.90\,(0.34)\times 10^{-4}$ & $8.03\,(0.48)\times 10^{-14}$ & $3.27\,(0.20)\times 10^{-10}$ \\
 & BB-scale & $3.93\,(0.07)$ & $1.51\,(0.09)\times 10^{-3}$ & $6.89\,(0.35)\times 10^{-4}$ & $2.880\,(0.145)$ \\
\addlinespace
Uniform mixing, $d=1$ (2000) & Ordinary & $4.53\,(0.06)$ & $7.00\,(0.43)\times 10^{-4}$ & $0$ & $0$ \\
 & Balanced & $4.48\,(0.06)$ & $7.14\,(0.43)\times 10^{-4}$ & $8.95\,(0.48)\times 10^{-6}$ & $0.037\,(0.002)$ \\
 & Tiny & $4.53\,(0.06)$ & $7.00\,(0.43)\times 10^{-4}$ & $1.06\,(0.07)\times 10^{-13}$ & $4.20\,(0.25)\times 10^{-10}$ \\
 & BB-scale & $4.31\,(0.06)$ & $1.37\,(0.08)\times 10^{-3}$ & $6.86\,(0.38)\times 10^{-4}$ & $2.849\,(0.152)$ \\
\addlinespace
Four-point square, $d=2$ (1000) & Ordinary & $9.85\,(0.26)$ & $4.20\,(0.23)\times 10^{-3}$ & $0$ & $0$ \\
 & Balanced & $9.72\,(0.25)$ & $4.26\,(0.23)\times 10^{-3}$ & $2.82\,(0.15)\times 10^{-5}$ & $0.062\,(0.003)$ \\
 & Tiny & $9.85\,(0.26)$ & $4.20\,(0.23)\times 10^{-3}$ & $7.01\,(2.42)\times 10^{-12}$ & $8.39\,(0.51)\times 10^{-9}$ \\
 & BB-scale & $9.20\,(0.24)$ & $8.36\,(0.38)\times 10^{-3}$ & $3.53\,(0.18)\times 10^{-3}$ & $7.863\,(0.371)$ \\
\addlinespace
\bottomrule
\end{tabular}

\end{table}

\subsection{Designs and computation}\label{app:numerics-design}\label{app:numerics-computation}

The first design is the one-dimensional three-point mixture
\[
 G^*=0.25\delta_{-2.5}+0.50\delta_0+0.25\delta_{2.5},
 \qquad K=[-4,4].
\]
The second uses the uniform mixing distribution $G^*=\operatorname{Unif}[-2.5,2.5]$ on the same parameter set, which checks that the conclusions do not depend on finite true support.  The third is two-dimensional, with equal mass at the four corners $(\pm1.75,\pm1.75)$ and $K=[-3.5,3.5]^2$.  The Monte Carlo study uses $n\in\{200,500,1000,2000\}$ with 100 replications for each one-dimensional design, and $n\in\{250,500,1000\}$ with 40 replications for the two-dimensional design.  On each dataset, the ordinary NPMLE and the three reweighted NPMLEs are computed from the same observations, with independent weight streams for the three concentrations, so that all comparisons are paired.

The regularization path uses a single sample of size $n=500$ from the three-point design and 16 concentrations: 14 logarithmically spaced values from 1 to $n^3$, the balanced concentration $\alpha_n=\numPathBalancedAlpha$, and ten times the balanced concentration.  Thirty independent weight vectors are drawn at each concentration, giving 480 weighted fits.  In every experiment, raw weights are drawn from $\mathrm{Gamma}(\alpha,\alpha)$ and normalized to $\widetilde W_i=nW_i/\sum_jW_j$, which have sample mean one and leave the maximizers unchanged.

All estimators are computed by fully corrective column generation over the parameter set $K$.  Given weights $w$ and a current mixing distribution $G$, the optimality residual is
\begin{equation}\label{eq:numerical-kkt}
 \mathcal R_w(G)=
 \sup_{\theta\in K}\left[
 \frac{1}{\sum_iw_i}\sum_{i=1}^n
 w_i\frac{\phi_d(X_i-\theta)}{f_G(X_i)}-1
 \right],
\end{equation}
which is nonnegative and equals zero exactly when $G$ maximizes the weighted likelihood.  At each iteration, an approximate maximizer of the bracket in \eqref{eq:numerical-kkt} is added to the support and all masses are reoptimized.  The supremum is approximated on a mesh of 1601 points refined at its local maxima in one dimension, and on a $21\times21$ mesh with 12 continuous local searches in two dimensions.  After column generation, atom locations and masses are polished jointly and the residual is reevaluated on a denser search that includes all active support locations.  A fit is accepted only if its residual is at most $5\times10^{-7}$ in one dimension or $2\times10^{-6}$ in two dimensions; otherwise the search is repeated on denser meshes.  During optimization, atoms of negligible mass are removed, and atoms whose locations nearly coincide are merged, with tolerances no larger than the resolution used to count support.  These searches approximate the supremum in \eqref{eq:numerical-kkt} but do not certify a global upper bound.

Every weighted fit is initialized at the computed ordinary NPMLE for the same dataset.  When $\alpha$ is very large, as for the Tiny fits, the ordinary NPMLE already nearly satisfies the weighted optimality condition, so differences between the two fits that are smaller than the optimization accuracy are not resolved.  Data and weight seeds are deterministic functions of the design, sample size, replication and concentration, so every reported number can be regenerated.  The computations used Python 3.13.9 with NumPy 2.3.5 and SciPy 1.16.3.

\subsection{Performance measures and accuracy}\label{app:numerics-measures}\label{app:numerics-accuracy}

Squared Hellinger distances, with the normalization of \Cref{sec:model}, are computed by the trapezoidal rule on $[-10,10]$ with 6001 points in one dimension and on $[-12,12]^2$ with a $211\times211$ grid in two dimensions; the true density of the uniform design is evaluated in closed form.  The resolved support size joins positive-mass atoms whose locations are closer than $0.01$ in one dimension or $0.05$ in two dimensions and counts the resulting connected clusters; it is a numerical resolution convention rather than an exact count of atoms in floating-point arithmetic.  The Wasserstein distance in \Cref{sec:numerics} is the Wasserstein-1 distance between the fitted mixing distributions, and path slopes are least-squares slopes of the logarithm of the medians on $\log\alpha$ over the concentrations $\alpha\ge10$.  The draw displayed in \Cref{fig:representative} for each concentration is the one whose squared Hellinger distance to the ordinary NPMLE is closest to the median at that concentration.  Means, standard errors, percentages, quantiles and regression coefficients are reported at the displayed precision, and ranges across design--sample-size cells are the rounded minimum and maximum of the cell summaries.

Across the \numReportedFits{} recorded fits (480 weighted path fits and 3680 Monte Carlo fits, the latter including the ordinary fits), the largest recorded residual is at most $\numMaxResidual$, a conservatively upward-rounded bound; the one-dimensional maximum rounds to $\numMaxResidualOneD$.  For one dataset of size 500 from each of the three-point and two-dimensional designs, with weights at the balanced concentration, ten dispersed random initializations gave the same resolved support size and largest squared Hellinger distances from the reference fit of $\numMultistartDistanceOneD$ and $\numMultistartDistanceTwoD$, respectively.  An independent check recomputes the saved likelihoods, compares the recorded residuals with denser searches, and compares the Hellinger integrals with wider and finer grids.  The signed ordinary log-likelihood gap is retained; the figures and \Cref{tab:numerical-summary} display its nonnegative part.  There are \numNegativeRawLossCount{} negative differences across the recorded fits, with minimum $\numMinRawLoss$; they reflect optimization error, so gaps of this order or smaller, including all gaps of the Tiny fits, are numerically unresolved.  These checks support the numerical comparisons but do not certify a continuous global optimum or establish uniqueness.

\subsection{Additional results}\label{app:numerics-results}\label{app:numerics-montecarlo}

On the regularization path, at the balanced concentration, \numPathBalancedSupportSummary{}; at $\alpha=n^3$, \numPathTinySupportSummary{}; and at $\alpha=1$, the 90th percentile of the resolved support size is $\numPathBBUpperSupport$.  At $\alpha=n^3$, the median squared Hellinger distance to the ordinary NPMLE is $\numPathTinyDistance$ and the median ordinary log-likelihood gap is $\numPathTinyLoss$.

\begin{figure}[htb]
\centering
\begin{minipage}[t]{0.28\textwidth}
  \centering\includegraphics[width=\linewidth]{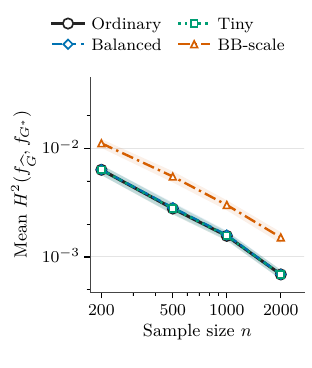}
  \par\vspace{-0.5em}\centerline{\small (a) Three-point, $d=1$}
\end{minipage}\hspace{0.03\textwidth}
\begin{minipage}[t]{0.28\textwidth}
  \centering\includegraphics[width=\linewidth]{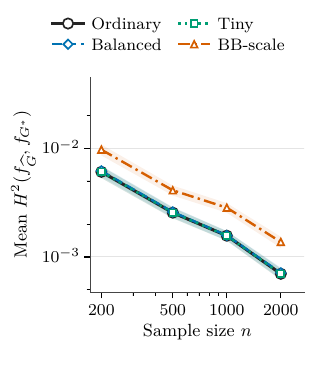}
  \par\vspace{-0.5em}\centerline{\small (b) Uniform mixing, $d=1$}
\end{minipage}\hspace{0.03\textwidth}
\begin{minipage}[t]{0.28\textwidth}
  \centering\includegraphics[width=\linewidth]{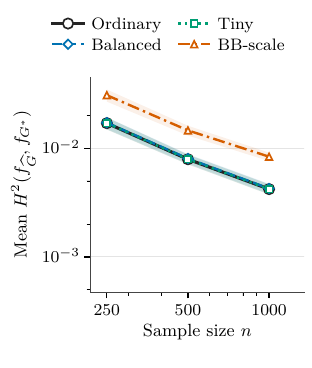}
  \par\vspace{-0.5em}\centerline{\small (c) Four-point square, $d=2$}
\end{minipage}
\caption{Estimated Hellinger risk, the Monte Carlo mean of the squared Hellinger loss, as a function of sample size.  Shaded bands show pointwise Monte Carlo means plus or minus 1.96 standard errors.}
\label{fig:risk}
\end{figure}

In the Monte Carlo study, \Cref{fig:risk} shows the estimated Hellinger risk across sample sizes, and \Cref{tab:numerical-summary} reports Monte Carlo means and standard errors at the largest sample size in each design.

The Balanced fits track the ordinary NPMLE closely in every design, consistent with \Cref{thm:main}(iii) and (iv).  Across the eleven design--sample-size cells, the ratio of their estimated Hellinger risk to that of the ordinary NPMLE ranges from $\numBalancedRiskMin$ to $\numBalancedRiskMax$, with average $\numBalancedRiskAverage$, and the mean squared Hellinger distance between the two fitted densities ranges from $\numBalancedDistanceMin$ to $\numBalancedDistanceMax$.  The resolved support sizes of the two estimators agree in $\numBalancedSupportAgreement\%$ of the \numPairedDatasets{} paired replications, and the mean difference in resolved support size lies between $\numBalancedSupportDifferenceMin$ and $\numBalancedSupportDifferenceMax$ in every cell, so reweighting does not inflate the support in these designs; its contribution is a unique sparse selection with a guarantee, not fewer atoms than an ordinary NPMLE that is already sparse.

The ordinary log-likelihood gap of the Balanced fits follows the scale of \Cref{thm:main}(iii): its cell means range from $\numBalancedLossMin$ to $\numBalancedLossMax$, and after multiplication by $\log n$ from $\numBalancedScaledLossMin$ to $\numBalancedScaledLossMax$; after the same scaling, those of the log-fit discrepancy range from $\numBalancedScaledLogfitMin$ to $\numBalancedScaledLogfitMax$, stably across sample sizes and dimensions.

The Tiny fits are numerically indistinguishable from the ordinary NPMLE: their ratio of estimated Hellinger risks lies between $\numTinyRiskMin$ and $\numTinyRiskMax$, their mean squared Hellinger distance to the ordinary NPMLE ranges from $\numTinyDistanceMin$ to $\numTinyDistanceMax$, and their resolved support size agrees with that of the ordinary NPMLE in $\numTinySupportAgreement\%$ of the replications.  Because these differences lie below the optimization accuracy and every weighted fit is initialized at the computed ordinary NPMLE (\Cref{app:numerics-computation}), this experiment cannot separate the exact reweighted maximizer from the ordinary one; the sparsity guarantee of \Cref{cor:tiny} concerns exact maximizers.

The BB-scale fits behave differently: their ratio of estimated Hellinger risks ranges from $\numBBRiskMin$ to $\numBBRiskMax$, and their mean ordinary log-likelihood gap from $\numBBLossMin$ to $\numBBLossMax$; they sometimes have a slightly smaller resolved support size, but at a material statistical cost.  This is the behavior expected outside the theory, since with $\alpha=1$ the weights do not concentrate and the reweighted objective does not approach the ordinary one.

\end{document}